\documentclass[12pt]{article}
\usepackage{bbm}
\usepackage{amsfonts}
\usepackage{amsmath}
\usepackage{latexsym}
\usepackage{natbib}
\usepackage{listings}
\usepackage{lscape}
\usepackage{graphicx}
\usepackage{setspace}
\usepackage{color}
\usepackage{multirow}
\usepackage{rotating}
\usepackage{epsfig}
\usepackage{booktabs}
\usepackage{tikz}
\usetikzlibrary{shapes,arrows}
\usepackage{topcapt}
\usepackage{lscape}

\usepackage{amsthm}

\newtheorem{definition}{Definition}
\newtheorem{theorem}{Theorem}
\newtheorem{lemma}{Lemma}

\newtheorem{condition}{Condition}

\def\sgn{\mathrm{sgn}}
\def\STGP{\mathcal{STGP}}

\newcommand{\beq}{ \begin{equation}}
\newcommand{\eeq}{ \end{equation}}
\newcommand{\beqn}{ \begin{eqnarray}}
\newcommand{\eeqn}{ \end{eqnarray}}

\newcommand{\bbeta}{ \mbox{\boldmath $\beta$}}

\newcommand{\bs}{ \mbox{\bf s}}

\newcommand{\calB}{{\cal B}}

\newcommand{\mR}{{\cal R}}
\newcommand{\mC}{{\cal C}}

\def\sgn{\mathrm{sgn}}
\def\STGP{\mathcal{STGP}}

\title{Scalar-on-Image Regression via the Soft-Thresholded Gaussian Process}
\author{~Jian~Kang\thanks{Jian~Kang is Assistant Professor in the Department of Biostatistics, University of Michigan, Ann Arbor, MI 48019. Email: jiankang@umich.edu.  Brian J. Reich is Associated Professor Department of Statistics, North Carolina State University, Raleigh, NC 27695. Email: bjreich@ncsu.edu. Ana-Maria Staicu is Associated Professor Department of Statistics, North Carolina State University, Raleigh, NC 27695. Email: astaicu@ncsu.edu},~Brian J. Reich~and~Ana-Maria Staicu}
\date{}                                           % Activate to display a given date or no date

\newcommand{\cm}[1]{\ignorespaces}

\def\bfa{\mathrm a}

\def\bfz{\mathrm z}
\def\bfs{\mathrm s}

\def\bfv{\mathrm v}

\def\bfh{\mathrm h}

\def\bft{\mathrm t}

\def\bfA{\mathrm A}

\def\bfS{\mathrm S}

\def\bfX{\mathrm X}
\def\bfY{\mathrm Y}
\def\bfZ{\mathrm Z}
\def\bfD{\mathrm D}

\def\bfM{\mathrm M}

\def\bfW{\mathrm W}

\def\bfeta{\overline \eta}

\def\bfK{\mathrm K}

\def\bftau{\overline \tau}

\def\bfalpha{\alpha^{\mathrm{v}}}
\def\bfbeta{\beta^{\mathrm{v}}}

\def\bfSigma{\Sigma}

\def\wtbeta{\tilde\beta^{\mathrm{v}}}

\def\cA{\mathcal A}
\def\cB{\mathcal B}

\def\cR{\mathcal{R}}
\def\cD{\mathcal D}

\def\mV{\mathrm{V}}

\def\cGP{\cG\cP}
\def\cF{\mathcal F}

\def\mVar{\mathrm{var}}

\def\mE{E}

\def\mbR{\mathbb R}

\def\md{\mathrm d}

\def\mN{\mathrm{N}}

\def\cGP{\mathcal{GP}}

\def\bbeta{\boldsymbol \beta}

\def\bs{\boldsymbol s}

\def\cA{\mathcal{A}}
\def\cGP{\mathcal{GP}}

\def\rT{\mathrm T}

\begin{document}

\bibliographystyle{asa}
\bibpunct{(}{)}{,}{a}{}{;}

\maketitle
\begin{abstract}
The focus of this work is on spatial variable selection for scalar-on-image regression. We propose a new class of Bayesian nonparametric models,  soft-thresholded Gaussian processes and develop the efficient posterior computation algorithms. Theoretically, soft-thresholded Gaussian processes provide large prior support for the spatially varying coefficients that enjoy piecewise smoothness, sparsity and continuity, characterizing the important features of imaging data. Also, under some mild regularity conditions, the soft-thresholded Gaussian process leads to the posterior consistency for both parameter estimation and variable selection for scalar-on-image regression,  even when the number of true predictors is larger than the sample size. The proposed method is illustrated via simulations, compared numerically with existing alternatives and applied to Electroencephalography (EEG) study of alcoholism.

\medskip


\end{abstract}

\newpage
\setlength{\baselineskip}{24pt}		
\section{Introduction}\label{s:intro}

Scalar-on-image regression has attracted considerable attention recently in both frequentist and Bayesian literature. This problem is challenging for several reasons such as: 1) the predictor is very high dimensional (two dimensional or three dimensional image), often larger than the sample size, 2) the observed predictor may be contaminated with noise and the true predictor signal may exhibit complex correlation structure, and 3) most components of the predictor may have no effect on the response, and when they have an effect it may vary smoothly. 

Regularized regression techniques are usually needed when the dimension of the predictor is much higher relative to the sample size; lasso~\citep{tibshirani1996regression} is a popular method for variable selection by employing a penalty based on the sum of the absolute values of the regression coefficients. However most approaches do not accommodate predictors with ordered components such as in the case of predictor images. One exception is the fused lasso, which generalizes the lasso by penalizing both the coefficients and their successive differences, thus ensuring both sparsity and smoothness of the effect. To incorporate spatial dependence structure of the predictors,  \cite{reiss2010functional} extended the functional principal component regression originally proposed for one dimensional functional covariates to high dimensional predictors. They modeled the coefficient function by using B-spline functions and considered common smoothing spline penalty which is not sensitive to sharp edges and jumps. Recently, \cite{WangZhu2015} proposed a new type of penalty - based on the total variation of the function - which yields piecewise smooth regression coefficients. While these approaches are computationally efficient, none of them can fully take into account the spatial dependence of the image predictor. In addition, in this framework it is not clear how to assess statistical significance.   

To overcome some of these limitations, this problem has been also approached form a Bayesian viewpoint. %\cite{huang2013bayesian} and 
\cite{Goldsmith:2014} proposed two latent spatial processes to model the sparsity and the smoothness of the regression coefficient: specifically an Ising prior was used for the binary indicator variable that controls whether a voxel image is predictive of the response or not (sparsity), and a conditional autoregressive Gaussian process for the non-zero regression coefficients to improve the model prediction (smoothness). The use of Ising prior for the binary indicator was first discussed in \cite{smith2007spatial} in the context of high dimensional predictors and was also recently exploited by \cite{li2014spatial} who proposed it jointly with a Dirichlet process prior. To address the computational challenge of a non-closed form for the probability function, the latter work proposed an analytical approach to derive bounds for the hyperparameters. 
% GAP  
One of the characteristics of both  \cite{li2014spatial} and \cite{Goldsmith:2014} is that the sparsity and smoothness are controlled separately by two different spatial processes. As a result, the transition from zero-areas to non-zero neighboring areas in the regression coefficient may be very abrupt. This does not seem realistic for entire brain regions, where it is expected to see a gradual effect in contain brain areas on the response.  

We propose a novel approach to spatial variable selection in the scalar-on-image regression by modeling the regression coefficients through a soft-thresholding transformation of latent Gaussian processes, to which we refer as soft-thresholded Gaussian processes. The soft-thresholding function is well known as the solution for the lasso estimate when the design matrix is orthonormal~\citep{tibshirani1996regression}. The soft-thresholded Gaussian process leads to different model properties than the existent literature: in particular it ensures a gradual transition between the zero and non-zero effects of the neighboring locations. Theoretically, we can show that it provides a large support for the spatially varying coefficient function in the model that enjoys piecewise smoothness, sparsity and continuous properties. The idea is inspired from
\cite{Boehm:2015} who considered it as a regularization technique for spatial variable selection. This approach does not assign prior probability mass at zero for regression coefficients and it is not designed for the scalar-on-image regression.  The use of the soft-thresholded Gaussian process has an attractive computational advantage over the competing methods, where the use of Ising prior makes it impossible to have a closed form probability distribution function making the computations challenging.  In particular,  we consider a low-rank spatial model for the latent process, which is important for the scalability of the method to large datasets. 
For theoretical results, in addition to the large support, we also can show that the soft-thresholded Gaussian process leads to the posterior consistency for both parameter estimation and variable selection for scalar-on-image regression.  That is, the posterior distribution of the spatially varying coefficient function concentrates in a small neighborhood of the true value and the its sign is also consistent with the true value with probability one as the number of subjects goes to infinity. These two results only need a few mild regularity conditions; the conclusions hold even when the number of true predictors is larger than the sample size.

%We develop estimation of the model parameters and discuss the theoretical properties of our estimators; in particular we focus on the posterior consistency for normal responses. The theoretical results extend \cite{li2010bayesian} to account for complex correlated predictors.

The proposed method is introduced for the case of single image predictor and Gaussian responses for simplicity. Nevertheless extensions to accommodate other type of covariates through a linear effect as well as generalized responses are straightforward. The methods are applied to the data from an electroencephalography (EEG) study of alcoholism (\underline{http://kdd.ics.uci.edu/datasets/eeg/eeg.data.html}), where of interest was to study the relation between the alcoholism status and the electrical brain activity over time. The data have been previously described in \cite{li2010dimension} and  \cite{zhou2014regularized} and consist of EEG signals received from 64 channel of electrodes located on subjects' scalp, corresponding to alcoholic subjects and healthy controls. The EEG signals are recorded for 256 seconds; leading to a high-dimensional predictor. Previous literature that analyzed these data ignored the spatial locations of the electrodes on the scalp, and thus considered a two-dimensional predictor. In contrast, we recover the locations of the electrodes from the standard electrode position nomenclature described by Fig. 1 of \underline{https://www.acns.org/pdf/guidelines/Guideline-5.pdf}, as shown in Fig. \ref{f:eeg}.  %shows the EEG signals for 60 channels corresponding to one random alcoholic subject and one healthy control. 
We study the same scientific question by accounting for the space-temporal dependence of the predictor.

\begin{figure}
	\caption{The standard electrode position nomenclature for 10-10 system}\label{f:eeg}
	\begin{center}
	\includegraphics[width=3in]{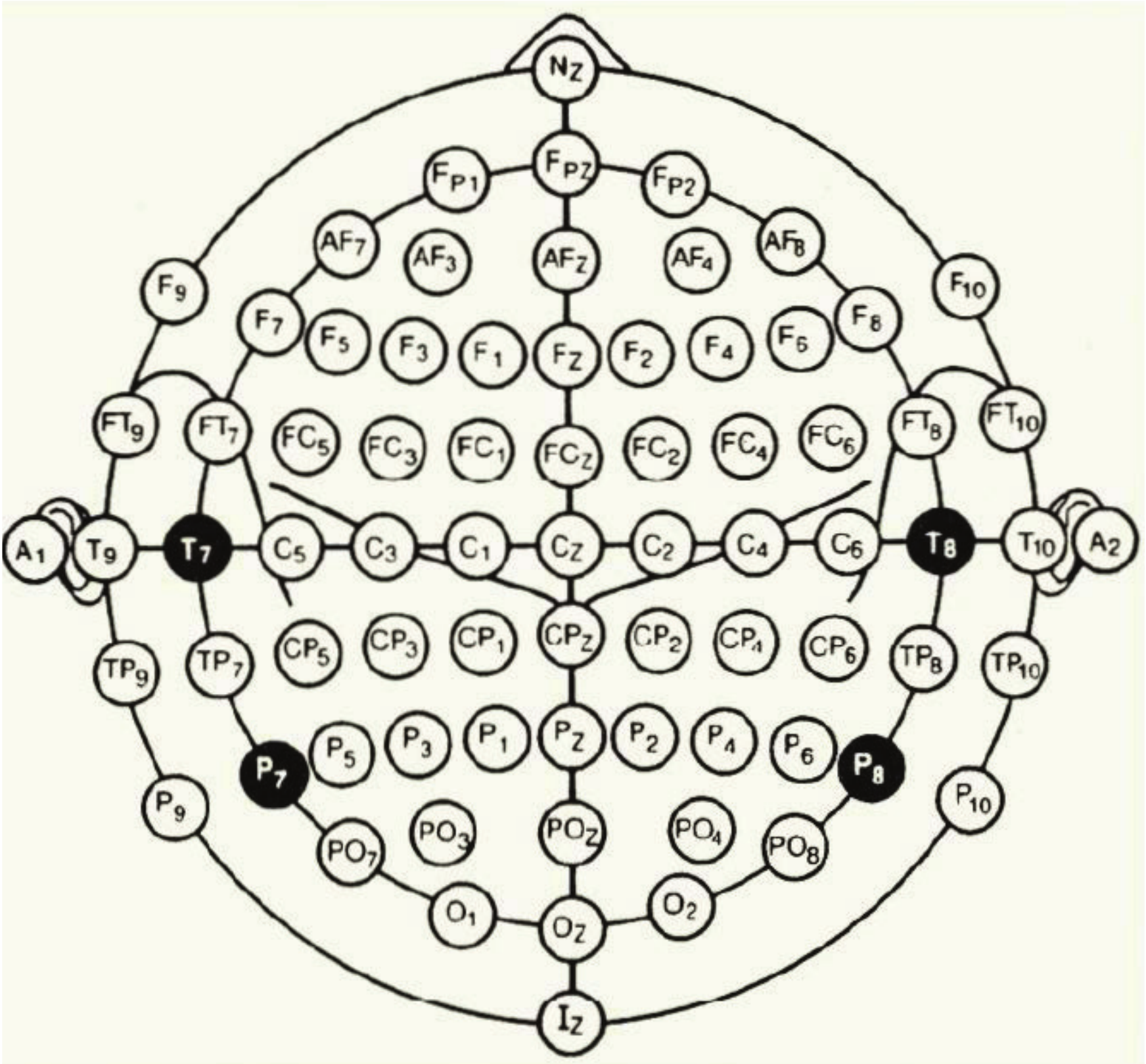}
     \end{center}
\end{figure}

%Could we do a 4 Dplot, where on the xy we have the spatial location of the brain, and then on Z we have the signal and we use color to depict the time dependence ? 

%The structure of the remainder paper is as follows. The proposed Bayesian modeling is described in Section 2[XX] and the theoretical properties of the estimators are presented in Section 3[XX]. Section 4[xx] details the implementation of the methodology. We investigate the proposed method in a simulation experiment and compare with its closest competitors in Section 5[xx]. The application to the EEG data is presented in Section 6[XX].

\section{Model}\label{s:model}
\subsection{Scalar-on-image regression}
Let $\mbR^d$ be a $d$-dimensional vector space of real values for any integer $d\geq 1$. Suppose there are $n$ subjects in the dataset. For each subject $i$, we collect a scalar response variable, $Y_i\in\mbR^1$, a set of $p_n$ spatially distributed imaging predictors, denoted $\bfX_i = (X_{i,1},\ldots, X_{i,p_n})^{\rT} \in \mbR^{p_n}$ and other scalar covariates $\bfW_i = (W_{i,1},\ldots, W_{i,q})^{\rT}\in\mbR^q$, where $X_{i,j}$ denotes the image intensity values measured at location $\bfs_j$, for $j = 1,\ldots, p_n$. Write $\bfS = \{\bfs_j\}_{j=1}^{p_n}$ which is an subset of a compact closed region $\cB\subseteq \mbR^d$. Let $\mN(\mu,\Sigma)$ be a normal distribution with mean $\mu$ and covariance $\Sigma$ (or variance for one dimensional case).  We consider the following scalar-on-image regression model:  for $i = 1,\ldots, n$, 
\begin{eqnarray}\label{eq:main_model}
\ \{ Y_i \mid \bfW_i, \bfX_{i},\bfalpha, \beta,\sigma^2, \bfS\ \}  \sim\mN\left\{\sum_{k=1}^q \alpha_{k} W_{i,k}+p_n^{-1/2}\sum_{i=1}^{p_n}\beta(\bfs_j) X_{i,j}, \sigma^2\right\},
\end{eqnarray}
where $\bfalpha = (\alpha_{1},\ldots,\alpha_{q})^{\rT}$ is the coefficient for the scalar covariates $\bfW_i$ and $\beta(\bfs)$ is a spatially varying coefficient function defined on  $\cB$ for imaging predictor $\bfX_i$. Assume that $\{\bfW_i\}_{i=1}^n$ are fixed design covariates and $\bfS$ collects all fixed spatial locations. In practice, the normalizing scalar $p^{-1/2}_n$ can be absorbed into imaging predictors; its role is to rescale the total effects of massive imaging predictors such that they are bounded away from infinity with a large probability when $p_n$ is very large. Scientifically, imaging predictors take values that measure the body tissue contrast or the neural activities at each spatial location and the number of imaging predictors, $p_n$ is determined by the resolution of the image. Thus, the total effects of imaging predictors reflect the total amount of the intensity in the brain signals, which should not increase to infinity as the image resolution increases. In model \eqref{eq:main_model}, the response is taken to be Gaussian and only one type of imaging predictor is included, although extensions to non-Gaussian responses and multi-modality imaging predictor regression are straightforward. 
 
\subsection{Soft-thresholded Gaussian Processes}
In order to capture the characteristics of imaging predictors and their effects on the response variable,  the prior for the covariate function $\beta$ should be sparse and spatial.  That is, we assume that many locations have $\beta(\bfs_j)=0$; the sites with non-zero coefficients cluster spatially, and that the coefficients vary smoothly in clusters of non-zero coefficients.  To encode these desired properties into the prior, we represent $\beta(\bfs)$ as a transformation of a Gaussian process, $\beta(\bfs) = g_{\lambda}\{\tilde \beta(\bfs)\}$, where $g_\lambda$ is the transformation function dependent on parameter $\lambda$ and $\tilde \beta(\bfs)$ follows a Gaussian process prior.  In this transkriging \citep{cressie:1993} or Gaussian copula \citep{nelsen:1999} model, the function $g_\lambda$ determines the marginal distribution of $\beta(\bfs)$, while the covariance of the latent $\tilde\beta(\bfs)$ determines $\beta(\bfs)$'s dependence structure.  

Spatial dependence is determined by the prior for $\tilde \beta(\bfs)$.  We assume that $\tilde \beta(\bfs)$ is a Gaussian process with  zero-mean and stationary covariance function $\mathrm{cov}\{\tilde \beta(\bfs), \tilde \beta (\bfs')\} = \kappa(\bfs-\bfs')$ for some covariance function $\kappa$.  Although other transformations are possible \citep{Boehm:2015}, we select $g_\lambda$ to be the soft-thresholding function to map $\tilde \beta(\bfs)$ near zero to exact zero and thus give a sparse prior.  Let
\beq\label{beta}
   g_\lambda(x) = \left\{\begin{array}{cl}
    0 & |x|\leq \lambda\\
   	\sgn(x)(|x|-\lambda) & |x|>\lambda
   \end{array}\right.,
\eeq
where $\sgn(x)$ is the sign of $x$, i.e. $\sgn(x) = 1$ if $x>0$ and $\sgn(x) = -1$ if $x<0$ and $\sgn(0) = 0$. The thresholding parameter $\lambda>0$ determines the degree of sparsity.  This soft-thresholded Gaussian process prior is denoted
$\beta \sim \STGP(\lambda,\kappa)$.

\section{Theoretical Properties}
In this section, we examine the theoretical properties of soft thresholded Gaussian processes  as prior models for scalar-on-image regression. We first introduce the formal definition for the class of the true spatially varying coefficients in the model that can well characterize the effects of the imaging predictors on the response variable. In light of good properties of the soft thresholding function (Lemmas \ref{lemma:Lipschitz} -- \ref{lemma:latent_process}),  we show that the soft thresholded Gaussian processes have large support for the true spatially varying coefficient functions in the class we define (Theorem \ref{largsup}). Then  we can verify the prior positivity of neighborhoods  (Lemma \ref{lem:prior_pos}) and construct uniform consistent tests (Lemma \ref{lem:uniform_test}) which are needed to define a sieve of spatially varying coefficient functions and find the upper bound of the tail probability (Lemmas A.\ref{lem:covering_number}--A.\ref{lem:tail_prob}), and verify that the model is identifiable (Lemmas A.\ref{lem:gamma_low_bound}--A.\ref{lem:test}) under certain regularity conditions. Thus, following the theory developed by \cite{choudhuri2004bayesian}, we show the posterior consistency (Theorem \ref{thm:post_consistency}). Given the smoothness and sparsity of the soft-thresholded Gaussian process, we can further show the posterior sign consistency for the sparse spatially varying coefficient function (Theorem \ref{thm:sign_consistency}), indicating the posterior spatial variable selection consistency. 
\subsection{Notations and Definitions}
We start with additional notations for the theoretical development and the formal definitions of the class of spatially varying coefficient functions under consideration. We assume that all the random variables and stochastic processes that we introduced in this article are  defined in the same probability space, denoted $(\Omega, \cF, \Pi)$.  Recall that $\mbR^d$ represents the $d$-dimensional vector space of real values.  Let  $\mathbb{Z}^d_+ = \{0,1,\ldots, \}^d \subset \mbR^d$ represent the $d$-dimensional vector space of non-negative integers. For any vector $\bfv  = (v_1,\ldots, v_d)^{\rT}\in \mbR^d$, let $\|\bfv\|_p = (\sum_{l=1}^d |v_l|^p)^{1/p}$ be the $L^p$ norm for vector $\bfv$ for any $p\geq 1$, and $\|\bfv\|_\infty = \max_{l=1}^d |v_l|$ be the supremum norm. For any $x\in \mbR$, let $\lceil x \rceil$ be the smallest integer larger than $x$ and $\lfloor x \rfloor $ be the largest integer smaller than $x$. Define the event indicator $I(\cA)\in\{0,1\}$ with $I[\cA]=1$ if event $\cA$ occurs, $I[\cA] = 0$ otherwise.  For any $\bfz = (z_1,\ldots, z_d)^{\rT} \in \mathbb{Z}^d_+$, define $\bfz !  = \prod_{l=1}^d \prod_{k=1}^{z_l} k $ and define $\bfv^{\bfz} = \prod_{l=1}^d v_l^{z_l}$. 
\begin{definition}\label{def:diff_function}
Let $f(\bfs)$ be defined in the set $\cB\subseteq \mbR^d$ for $\bfs = (s_1,\ldots, s_d)$, and let $m$ be a non-negative integer.  We say $f(\bfs)$ is a differentiable function of order $m$, if $f(\bfs)$ has partial derivatives 
$$D^{\bftau} f(\bfs) = \frac{\partial^{\|\bftau\|_1}f}{s_1^{\tau_1}\cdots s_d^{\tau_d}}(\bfs) = \sum_{\|\bfeta\|_1 + \|\bftau\|_1\leq m} \frac{D^{\bftau+\bfeta}f(\bft)}{\bfeta!}(\bfs-\bft)^{\bfeta} + R_m(\bfs,\bft),$$ 
where $\bftau = (\tau_1, \ldots, \tau_d)^{\rT}\in \mathbb {Z}_+^{d}$,  $\overline\eta\in\mathbb{Z}_+^{d}$ and $\bft \in \mbR^d$.
\end{definition}
 Denote by $\mC^{m}(\cB)$ a set of differentiable functions of order $m$ defined on $\cB$. For any $f\in \mC^{m}(\cB)$, define the $L^p$ norm $\|f\|_p = \left(\int_{\cB}|f(\bfs)|^p d \bfs\right)^{1/p}$ for any $p\geq 1$ and the supremum norm is $\|f\|_{\infty} = \sup_{\bfs\in\cB}|f(\bfs)|$. The reminder $R_m(\bfs,\bft)$ has the following property. Given any point $\bfs_0$ of $\cB$ and any $\epsilon >0$, there is a $\delta>0$ such that if $\bfs$ and $\bft$ are any two points of $\cB$ with $\|\bfs-\bfs_0\|_1<\delta$ and $\|\bft-\bfs_0\|_1<\delta$, then  $|R_m(\bfs,\bft)|\leq \|\bfs-\bft\|_1^{m-\|\bftau\|_1}\epsilon$. If $\|D^{\bftau}f\|_\infty \leq M <\infty$, then $|R_m(\bfs,\bft)|\leq (M\|\bfs-\bft\|_1^{m+1})/(m+1)!$. 

\begin{definition}\label{def:beta}
Denote by $\overline\cR$ and $\partial \cR$ the closure and the boundary of any set $\cR\subseteq \cB$. Define a collection of functions $\beta(\bfs)$ that satisfy the following properties. That is,  there exist two disjoint non-empty open sets $\cR_{-1}$ and $\cR_{1}$ with $\overline\cR_1 \cap \overline \cR_{-1} = \emptyset$ such that 
\begin{enumerate}
 	\item[](\ref{cond:beta}.1)  Piecewise Smoothness: $\beta(\bfs)$ is smooth over $\overline\cR_{-1}\cup \overline\cR_{1}$,  i.e.
	$$\beta(\bfs)I[\bfs\in \overline\cR_{-1} \cup \overline\cR_{1}]\in \mC^\rho(\overline\cR_{-1}\cup \overline \cR_{1}), \mbox{ with } \rho = \lceil d/2 \rceil.$$
		\item[](\ref{cond:beta}.2) Sparsity:  $\beta(\bfs)=0$ for $\bfs\in \cR_0  $, $\beta(\bfs)>0$ for $\bfs\in \cR_1$ and $\beta(\bfs)<0$ for $\bfs\in \cR_{-1}$, where $\cR_0  = \cB - (\cR_{-1}\cup \cR_1)$ and $\cR_0-(\partial \cR_1 \cup \partial \cR_{-1})\neq \emptyset$. 
		\item[](\ref{cond:beta}.3) Continuity: $\beta(\bfs)$ is continuous over  $\cB$. That is, 
	$$\lim_{\bfs\rightarrow \bfs_0}\beta(\bfs) = \beta(\bfs_0), \qquad\mbox{ for any } \bfs_0 \in \cB.$$
\end{enumerate}
\end{definition}
Define $\Theta$ as a collection of all spatially varying coefficient functions that satisfy with conditions (\ref{cond:beta}.1) -- (\ref{cond:beta}.3) in Definition~\ref{def:beta}. 

\subsection{Conditions for Theoretical Results}
In this section, we list all the conditions that are needed to facilitate the theoretical results, although they may not  be the weakest conditions. 
\begin{condition}\label{cond:num_cov_order}
There exists $M_0>0$, $M_1>0$, $N\geq 1$, and some $\upsilon_0$, $d/(2\rho)<\upsilon_0<1$ and $\rho = \lceil d/2\rceil$ such that for all $n > N$,  $M_0 n^d \leq p_n \leq M_1 n^{2\rho \upsilon_0}$. 
\end{condition}

This condition implies that the number of imaging predictors $p_n$ should be of polynomial order of sample size. The lower bound indicates that $p_n$ needs to be sufficiently large such that the posterior distribution of the spatially varying coefficient function concentrates around the true value. 

\begin{condition}\label{cond:beta}
 The true spatially varying coefficient function in model \eqref{eq:main_model} enjoys the piecewise smoothness, sparsity and continuity properties, in short, $\beta_0 \in\Theta$. 
\end{condition}

The next two conditions summarize constraints on the spatial locations and the distribution of the imaging predictors.
\begin{condition}\label{cond:loc}
For the observed spatial locations $\bfS = \{\bfs_j\}_{j=1}^{p_n}$ in region $\cB$, there exists a set of sub-regions $\{\cB_{j}\}_{j=1}^{p_n}$ satisfying the following conditions:
\begin{enumerate}
\item[](\ref{cond:loc}.1) They form a partition of $\cB$, i.e. $\cB = \bigcup_{j=1}^{p_n} \cB_{j}$ with $\cB_{j}\cap \cB_{j'} = \emptyset$.
\item[](\ref{cond:loc}.2) For each $j = 1,\ldots, p_n$,  $\bfs_j \in \cB_j$ and $\mV(\cB_j) \leq \zeta(\cB_j) < \infty$, where $V(\bullet)$ is the Lebesgue measure and 
$$\zeta(\cB) = \sup_{\bft, \bft'\in \cB} \left[\max_{k} |t_{k} - t'_k|\right ]^d,\quad \mbox{ with } \bft = (t_1,\ldots, t_d)^{\rT}\mbox{and } \bft' = (t'_1,\ldots, t'_d)^{\rT}.$$
\item[](\ref{cond:loc}.3)  There exists a constant $0<K< \mV(\cB)$ such that $\max_{j} \zeta(\cB_j) < 1/(K p_n)$ as $n\to\infty$.
%\item[(C2.4)] The number of observed spatial locations where the true SVCF has non-zero values is bounded.
%$$\sum_{j=1}^{p_n} I[\beta_0(\bfs_j) \neq 0 ] = \sum_{j=1}^{p_n} I[\bfs_j \in \cR_{-1}\cup \cR_{1} ] < \infty.$$
\end{enumerate}
\end{condition}

When $\cB$ is a hypercube in $\mR^d$, e.g. $\cB = [0,1]^d$, there exists a set of $\{\cB_j\}_{j=1}^{p_n}$ that equally partitions $\cB$. Then $\mV(\cB_j) = \zeta(\cB_j) = p^{-1}_{n}$. 
\def\mmin{\mathrm{min}}
\def\mmax{\mathrm{max}}

\begin{condition}\label{cond:cov}
The covariate variables $\{X_{i,1},\ldots, X_{i,p_n}\}_{i=1}^n$ are independent realizations of a stochastic process $X(\bfs)$ at spatial locations $\bfs_1,\ldots, \bfs_{p_n}$.  The stochastic process $X(\bfs)$ satisfies
\begin{enumerate}
\item[](\ref{cond:cov}.1) $E\left[X(\bfs)\right] = 0$ for all $\bfs \in \cB$. 
\item[](\ref{cond:cov}.2)  For all $n >1$, let $\bfSigma_n = (\sigma_{j,j'})_{1\leq j,j'\leq p_n}$ where  $\sigma_{j,j'}= E\left[X(\bfs_j)X(\bfs_{j'})\right]$. Let $\rho_{\mathrm{min}}(A)$ and $\rho_{\mathrm{max}}(A)$  be the smallest eigenvalue and the largest eigenvalue of a matrix $A$, respectively. Then there exist $c_\mmin$ and $c_\mmax$ with $0<c_\mmin\leq1$ and $0 <  c_\mmax < \infty$ such that for $n>1$,
$$\rho_{\mmin} \left(\bfSigma_n \right)  > c_\mmin \mbox{ and } \rho_{\mmax} \left(\bfSigma_n \right)  < c_\mmax.$$
\item[](\ref{cond:cov}.3)  For any $\epsilon>0$ and $M<\infty$, there exists $\delta > 0$ such that for any $a_1,\ldots, a_{p_n} \in \mathbb R$ with $|a_j|<M$ for all $j$, if there exits $N$, for all $n$, $p^{-1}\sum_{j=1}^{p_n}|a_j| > \epsilon,$ then 
$$\Pi\left[p_n^{-1/2}\left|\sum_{j=1}^{p_n} a_j X(\bfs_j)\right| > \delta \right] >  \delta.$$
\end{enumerate}
\end{condition}

Condition \ref{cond:cov} includes assumptions on the mean of $X(\bfs)$ and on the range of eigenvalues of the covariance matrix $\bfSigma_n$ for covariate variables. If Gaussian process $X(\bfs)$ on $[0,1]^d$  has zero mean and $\mE[X(\bfs_j)X(\bfs_{j'})] = \rho_0\exp(- p_{n}\|\bfs_j-\bfs_{j'}\|_1)$, $0<\rho_0<1$, if $j\neq j'$ and $\mE[X(\bfs_j)^2] = 1$, where $\{\bfs_j\}_{j=1}^{p_n}$ are chosen as the centers of the equal space partitions of $\cB$, then condition (\ref{cond:cov}.2) holds.  Furthermore, condition (\ref{cond:cov}.3) also holds. Specifically, for any $\epsilon > 0$, taking $\delta = c_\mmin^{1/2}\epsilon\exp(-\epsilon)$, for any $a_1,\ldots, a_{p_n}$, 
let $\xi = p_n^{-1/2}\sum_{j=1}^{p_n} a_j X(\bfs_j) \sim \mN(0, \kappa^2),$
By Condition (\ref{cond:cov}.2),
$$\kappa^2 = \frac{1}{p_n}\sum_{j,j'} a_j \sigma_{j,j'} a_{j'} \geq \frac{1}{p_n} \sum_{j=1}^{p_n} a_j^2 \rho_{\mmin}(\bfSigma_n) > c_{\mmin}\frac{1}{p_n} \sum_{j=1}^{p_n} a_j^2. $$
There exists $N$, such that for all $n>N$, $\left(p_n^{-1}\sum_{j=1}^{p_n} a_j^2\right)^{1/2} \geq p^{-1}_n \sum_{j=1}^{p_n} |a_j| > \epsilon. $
Thus, $\kappa^2>c_{\mmin}\epsilon^2$. Furthermore,
$$\Pi[|\xi| > \delta] =2\Phi\left(-\kappa^{-1}\delta\right) > 2\Phi\left(-c_{\mmin}^{-1/2}\epsilon^{-1}\delta\right) = 2\Phi\left\{-\exp(-\epsilon)\right\}>\epsilon\exp(-\epsilon)> \delta.$$

To ensure the large support property, we need the following condition on the kernel function of the Gaussian process. This condition also has been used previously by~\cite{ghosal2006posterior}. 
\begin{condition}\label{cond:kernel}
For every fixed $\bfs \in \cB$, the covariance kernel $\kappa(\bfs,\cdot)$ has continuous partial derivatives up to order $2\rho+2$. 
\end{condition}

\subsection{Large Support}
One of the desired properties for the Bayesian nonparametric model is to have  prior support over a large class of functions. In this section, we show that the support of the soft-thresholded Gaussian process is large for any spatially varying coefficient function of our interests in the scalar-on-image regression. We begin with two appealing properties of the soft-thresholding function in the following two lemmas. 
\begin{lemma}\label{lemma:Lipschitz}
The soft-thresholding function $g_\lambda(x)$ is Lipschitz continuous for any $\lambda > 0$, that is, for all $x_1, x_2 \in \mbR$, 
%\begin{eqnarray*}%\label{eqn:Lipschitz}
$|g_\lambda (x_1)-g_\lambda (x_2)| \leq |x_1 - x_2|.$
%\end{eqnarray*}
\end{lemma}

\begin{lemma}\label{lemma:latent_process}
For any function $\beta_0\in\Theta$, there exists a threshold parameter $\lambda_0$ and a smooth function $\tilde\beta_0(\bfs)\in \mC^{\rho}(\cB)$ such that 
%\begin{eqnarray*}%\label{eqn:latent_process}
$\beta_0(\bfs) = g_{\lambda_0}\{\tilde\beta_0(\bfs)\}. $
%\end{eqnarray*}
\end{lemma}

The proof of lemma \ref{lemma:Lipschitz} is straightforward by verifying the definition. The proof of lemma \ref{lemma:latent_process} is not trivial, it requires a detailed construction on the smooth function $\tilde\beta_0(\bfs)$. Please refer to the Appendix for details. 
\begin{theorem}\label{largsup}
(Large Support) For a function $\beta_0 \in \Theta$, there exists a thresholding parameter $\lambda_0$, such that the soft thresholded Gaussian process prior $\beta \sim\STGP(\lambda_0, \kappa)$ satisfies
	$$\Pi\left(\|\beta-\beta_0\|_{\infty}<\varepsilon\right)>0, \quad\quad \mbox{for all  } \varepsilon>0.$$
\end{theorem}
\begin{proof} 
By Lemma \ref{lemma:latent_process}, there exists a thresholding parameter $\lambda_0$ and a smooth function $\tilde\beta_0(\bfs)$ such that  $\beta_0(\bfs) = g_{\lambda_0}[\tilde\beta_0(\bfs)] $. Since $\beta\sim \STGP(\lambda_0,\kappa)$, we have $\beta(\bfs) = g_{\lambda_0}[\tilde\beta(\bfs)]$ with $\tilde\beta(\bfs) \sim \cGP(0, \kappa),$ By Lemma \ref{lemma:Lipschitz} and Theorem 4 of \citep{ghosal2006posterior}
\begin{eqnarray*}
\Pi\left(\sup_{\bfs\in \cB}|\beta(\bfs)-\beta_0(\bfs)|<\varepsilon\right)&=&\Pi\left(\sup_{\bfs\in \cB}|g_{\lambda_0}(\tilde\beta(\bfs))-g_{\lambda_0}(\tilde\beta_0(\bfs))|<\varepsilon\right)\\
&&\geq \Pi\left(\sup_{\bfs\in \cB}|\tilde\beta(\bfs)-\tilde\beta_0(\bfs)|<\varepsilon\right) > 0.
\end{eqnarray*}
\end{proof}
Theorem \ref{largsup} implies that there is always a positive probability that the soft-thresholded Gaussian process concentrates in an arbitrarily small neighborhood of any spatially varying coefficient function that has piecewise smoothness, sparsity and continuity properties. 
\subsection{Posterior Consistency}
For $i  = 1,\ldots, n$, given the image predictor $\bfX_i$ on a set of spatial locations $\bfS$ and other covariates $\bfW_i$, suppose the response $Y_i$ is generated from the scalar-on-image regression model \eqref{eq:main_model} with parameters $\bfalpha_0 \in \mbR^{q}$, $\sigma^2_0>0$ that are known and parameter of interest $\beta_0 \in\Theta$. The assumptions about $\bfalpha_0$ and $\sigma^2_0$ are for theoretical convenience; in practice it is straightforward to estimate from the data. We assign a soft-thresholded Gaussian process prior for the spatially varying coefficient function, i.e. $\beta \sim \STGP(\lambda_0, \kappa)$ for some known $\lambda_0>0$ and covariance kernel $\kappa$. In light of the large support by Theorem \ref{largsup}, the following lemma shows the positivity of prior neighborhoods:
\begin{lemma}\label{lem:prior_pos}
(Positivity of prior neighborhoods) Denote by $\pi_{n,i}(\bullet; \beta)$ the density function of $\bfZ_{n,i} = (Y_i,\bfX_{i})$ in model \eqref{eq:main_model} and suppose condition~\eqref{cond:cov} holds for $\bfX_i$. Define
\begin{eqnarray*}
\Lambda_{n,i}(\bullet; \beta_0,\beta) &=& \log \pi_{n,i}(\bullet; \beta) -  \log \pi_{n,i}(\bullet; \beta_0),\\
K_{n,i}(\beta_0,\beta) &=& \mE_{\beta_0} \{\Lambda_{n,i}(\bfZ_{n,i}; \beta_0,\beta)\}, \mbox{ and }\\
V_{n,i}(\beta_0,\beta) &=& \mVar_{\beta_0} \{\Lambda_{n,i}(\bfZ_{n,i}; \beta_0,\beta)\}.
\end{eqnarray*}
There exists a set $B$ with $\Pi(B)>0$ such that, for any $\epsilon > 0$, 
\begin{eqnarray*}
\liminf_{n\rightarrow\infty} \Pi\left[\left\{\beta\in B, \frac{1}{n}\sum_{i=1}^n K_{n,i}(\beta_0,\beta)<\epsilon\right\}\right] > 0, \quad\mbox{and}\quad \frac{1}{n^2} \sum_{i=1}^n V_{n,i}(\beta_0,\beta)  \rightarrow 0, \mbox{ for all } \beta \in B. 
\end{eqnarray*}
\end{lemma}

We construct sieves for the spatially varying coefficient functions in $\Theta$ as
$$\Theta_n = \left\{\beta \in\Theta: \|\beta\|_\infty \leq p_n^{1/(2d)}, \sup_{\bfs \in \cR_1\cup \cR_{-1}}|D^{\bftau}\beta(\bfs)| \leq p_n^{1/(2d)}, 1\leq\|\bftau\|_1 \leq \rho \right\},$$
and by Lemmas A.\ref{lem:covering_number} -- A.\ref{lem:test} in the appendix, we can find the upper bound of the tail probability and construct uniform consistent tests in the following lemmas:
 \begin{lemma}\label{lem:tail_prob}
 Suppose $\beta(\bfs)\sim  \STGP(\lambda_0,\kappa)$ with $\lambda_0>0$ and the kernel function $\kappa$ satisfies condition \ref{cond:kernel}, then there exist constants $K$ and $b$ such that for all $n\geq 1$, 
$$\Pi(\Theta_n^C) \leq K \exp(-b p_n^{1/d}).$$
\end{lemma}

\begin{lemma}\label{lem:uniform_test}
(Uniform consistent tests) For any $\epsilon>0$ and $\upsilon_0/2<\upsilon<1/2$, there exist $N$, $C_0$, $C_1$ and  $C_2$ such that for all $n>N$ and all $\beta \in \Theta_n$ , if $\|\beta - \beta_0\|_1>\epsilon$, a test function 
$\Psi_n$ can be constructed such that
\begin{eqnarray*}
\mE_{\beta_0}(\Psi_n)  \leq C_{0}\exp\left(-C_{2} n^{2\upsilon}\right), \mbox{ and }\quad  \mE_{\beta}(1-\Psi_n) &\leq&C_{0} \exp\left(-C_{1} n\right). 
\end{eqnarray*}
\end{lemma}

Proofs of Lemmas \ref{lem:prior_pos}--\ref{lem:uniform_test} are provided in the online supplementary materials. These lemmas verify three important conditions for proving posterior consistency in the scalar-on-image regression based on Theorem A.1 by~\cite{choudhuri2004bayesian}. Thus, we have the following theorem: 

\def\cU{\mathcal U}
\begin{theorem}\label{thm:post_consistency}
	(Posterior Consistency) Write data $\bfD_n = [\{Y_i\}_{i=1}^n, \{\bfX_i\}_{i=1}^n,\{\bfW_i\}_{i=1}^n]$. If Conditions 1 -- 5 hold,  then for any $\epsilon>0$, 
	$$ \Pi\left[\beta\in\Theta: \|\beta-\beta_0\|_1<\epsilon \mid\bfD_n \right] \to 1, $$
	as $n\to\infty$ in $P_{\beta_0}^n$ probability, where $P_{\beta_0}^n$ denotes the actual distribution of data $\bfD_n$. 
\end{theorem}
Theorem \ref{thm:post_consistency} implies that the soft-thresholded Gaussian process prior can ensure that the posterior distribution of the spatially varying coefficient function concentrates in an arbitrarily small neighborhood of the true value, when both the number of subjects and number of spatial locations are sufficiently large. Given that the true  function of interest is piecewise smooth, sparse and continuous, the soft-threshold Gaussian process prior can further ensure that the posterior probability of the sign of the spatially varying coefficient function being correct converges to one as the sample size goes to infinity. The result is formally stated in the following theorem. 
\begin{theorem}\label{thm:sign_consistency}
	(Posterior Sign Consistency) Suppose the model assumptions, prior settings and  regularity conditions for Theorem \ref{thm:post_consistency} hold. 
	$$ \Pi\left[ \sgn\{\beta(\bfs)\} = \sgn\{\beta_0(\bfs)\}, \mbox{ for all } \bfs\in \cB\mid \bfD_n\right] \to 1, $$
	as $n\to\infty$ in $P_{\beta_0}^n$ probability.
\end{theorem}
This theorem establishes the spatial variable selection consistency. It does not require the number of true imaging predictors is finite or less than the sample size.  This is different from most previous results, but it is reasonable in that the true spatially varying coefficient function is piecewise smooth and continuous and the soft-thresholded Gaussian process will borrow strength from neighboring locations to estimate the true imaging predictors. Please refer to the Appendix for the proofs of Theorems \ref{thm:post_consistency} and \ref{thm:sign_consistency}. 

\section{Posterior Computation}
\subsection{Model Representation and Prior Specifications}
We turn now to the practical applicability of our proposed method. We select a low-rank spatial model to ensure that computation remains possible for applications with large datasets.  We exploit the kernel convolution approximation of a spatial Gaussian process.  As discussed in \cite{higdon-1998}, any stationary Gaussian process $V(\bfs)$ can be written
%\beq\label{KC}
$V(\bfs) = \int K(\bfs_j-\bft)w(\bft)d\bft$,
%\eeq
where $K$ is a kernel function and $w$ is a white-noise process with mean zero and variance $\sigma_w^2$.  This gives covariance function
%\beq
$$   \mbox{cov}(\bfs,\bfs+\bfh) = \kappa(\bfh)=\sigma_a^2\int K(\bfs-\bft)K(\bfs+\bfh-\bft)d\bft, $$
%\eeq  
which illustrates the connection between covariance $\kappa$ and kernel $K$.  This representation suggests the approximation for the latent Gaussian process
%\beq
  $$ \tilde\beta(\bfs) = \sum_{l=1}^LK(\bfs-\bft_l)a_l,$$
%\eeq  
where $\bft_1,...,\bft_L\in\mbR^d$ are a grid of spatial knots covering $\calB$, $K$ is a local kernel function, and $a_l\sim\mbox{N}(0,\sigma_a^2)$ is the coefficient associated with knot $l$.  We use tapered Gaussian kernels with bandwidth $\sigma_h$,
%\beq\label{K}
$$K(h) = \exp\left[-\frac{h^2}{2\sigma_h^2}\right]I(h<3\sigma_h),$$
%\eeq
so that $K(||\bfs-\bft_l||)=0$ for $\bs$ separated from $\bft_l$ by at least $3\sigma_h$.  Taking $L<p$ knots and selecting compact kernels both lead to computational savings, as discussed in Section \ref{s:MCMC}.

The compact kernels $K$ control the local spatial structure and the prior for the coefficients $\bfa = (a_1,\ldots,a_L)^{\rT}$ controls broad spatial structure.  Following the work by \cite{Nychka:2015} for geostatistical data, we assume that the knots $\bft_1,\ldots,\bft_L$ are arranged on an $m_1 \times \cdots \times m_d$ array, and use $l\sim k$ to denote that knots $\bft_l$ and $\bft_k$ are adjacent on this array.  We then use a conditionally autoregressive prior \citep{handbook:2010} for the kernel coefficients.  The conditional autoregressive prior is also defined locally, with full conditional distribution
\beq\label{CR}
   a_l|a_k, k\ne l \sim \mbox{N}\left(\frac{\vartheta}{n_l}\sum_{k\sim l}a_k, \frac{\sigma_a^2}{n_l}\right),
\eeq    
where $n_l$ is the number of knots adjacent to knot $l$, $\vartheta\in(0,1)$ determines the strength of spatial dependence, and $\sigma_a^2$ determines the variance. These full conditional distributions correspond to the joint distribution $\bfa \sim$ N$[0,\sigma_a^2(\bfM-\vartheta \bfA)^{-1}]$, where $\bfM$ is diagonal with diagonal elements $\{n_1,...,n_L\}$ and $\bfA$ is the adjacency matrix with $(k,l)$ element equal 1 if $k\sim l$ and zero otherwise (including zeros on the diagonal). 

%To facilitate the posterior computation, we will introduce the matrix notation for model \eqref{eq:main_model}.  
Write $\wtbeta = \{\tilde\beta(\bfs_1),\ldots,\tilde\beta(\bfs_p)\}^{\rT}$. Denote by $\bfK$ the $p\times L$ kernel matrix with $(j,l)$ element $K(||\bfs_j-\bft_l||_{2})$, the prior for $\tilde\bfbeta$ is given by  $\wtbeta\sim \mN\{0,\sigma_a^2\bfK(\bfM-\vartheta\bfA)^{-1}\bfK^{\rT}\}$.  In this case, the $\tilde\beta(\bfs_j)$ do not have the equal variance, which may not generally be desirable.  Non-constant variance arises because the kernel knots $\bft_j$ may be unequally distributed, and because the conditional autoregressive model is non-stationary in that the variances of the $a_l$ are unequal.  

To stabilize the prior variance, define ${\tilde K}_{j,l} = K(||\bfs_j-\bft_l||_2)/w_{j}$ and ${\tilde \bfK}$ as the corresponding $p\times L$ matrix of standardized kernel coefficients, where $w_j$ are constants chosen so that the prior variance for each $\beta_j$ is equal.  We take $w_j$ to be the $j$-th diagonal element of $\mathrm{cov}(\tilde\beta^{\mathrm v})= \bfK(\bfM-\vartheta \bfA)^{-1}\bfK^T$, hence the kernel functions now depend on $\vartheta$.  By pulling the prior standard deviation  $\sigma_a$ out of the thresholding transformation, write $\bfbeta = \{\beta(\bfs_1),\ldots,\beta(\bfs_p)\}^{\rT}$, we have an equivalent model representation of model \eqref{eq:main_model} as
\beq\label{beta2}
Y_i \sim\mN[\bfW_i^{\rT}\bfalpha + p^{-1/2}_n\bfX_i^{\rT}\bfbeta,\sigma^2], \mbox{ with }  \beta(\bfs_j) = \sigma_a g_\lambda\{\tilde\beta(\bfs_j)\},
\eeq
where ${\wtbeta} \sim\mN\{0,{\tilde \bfK}(\bfM-\vartheta\bfA)^{-1}{\tilde \bfK^{\rT}}\}$. After standardization the prior variance of each $\tilde\beta(\bfs_j)$ is one, and therefore the prior probability that $\tilde \beta(\bfs_j)$ is nonzero is $2\Phi(-\lambda)$ for all $j$, where $\Phi(\bullet)$ denotes the cumulative distribution function of standard normal distribution. This endows each parameter with a distinct interpretation: $\sigma_a$ controls the scale of the non-zero coefficients; $\lambda$ controls the prior degree of sparsity; and $\vartheta$ controls spatial dependence. 

In practice, we normalize the response and covariates, and then select priors $\bfalpha\sim\mbox{N}(0,10^2 \mathrm{I}_q)$, $\sigma^2\sim\mbox{InvGamma}(0.1,0.1)$, $\sigma_a\sim \mbox{HalfNormal}(0,1)$, $\vartheta\sim\mbox{Beta}(10,1)$, and $\lambda\sim\mbox{Uniform}(\lambda_l,\lambda_u)$. Following \cite{banerjee:2004}, we use a beta prior for $\vartheta$ with mean near one because only values near one provide appreciable spatial dependence.  In many of the cases considered in the simulation studies, the sparsity parameter $\lambda$ cannot be fully identified.  To improve numerical stability, we suggest an informative data-driven prior.  We first fit the non-sparse model with $\lambda=0$ and record the proportion of the $\beta(\bfs_j)$ with posterior 95\% credible interval that exclude zero, denoted $u$.  The prior for $\lambda$ then restricts the prior proportion of non-zeros to be within 0.05 of $u$, i.e., $\lambda_l = -\Phi^{-1}[(u+0.05)/2]$ and $\lambda_u = -\Phi^{-1}[(u-0.05)/2]$.

\subsection{Markov chain Monte Carlo Algorithm}\label{s:MCMC}

We sample from the posterior distribution using Metropolis-Hastings within Gibbs sampling.  The parameters $\bfalpha$, $\sigma^2$, and $\sigma^2_a$ have conjugate full conditional distributions and are updated using Gibbs sampling.  The spatial dependence parameter $\vartheta$ is sampled using Metropolis-Hastings sampling using a beta candidate distribution with the current value as mean and standard deviation tuned to give acceptance around 0.4.  The threshold $\lambda$ is updated using Metropolis sampling with random-walk Gaussian candidate distribution with standard deviation tuned to have acceptance probability around 0.4.  The Metropolis update for $a_l$ uses the prior full conditional distribution in \eqref{CR} as the candidate distribution which gives high acceptance rate and thus good mixing without tuning.  

%{\color{blue} [[***Brian, do we need to provide a brief description of the MCMC algorithm here? Or can you simply write a few sentences say the standard MCMC algorithm can be applied? ***]]}
%

\section{Simulation study}\label{s:sim}

\subsection{Data generation}

In this section we conduct a simulation study to compare the proposed methods with other popular methods for scaler-on-image regression.  For each simulated observation, we generate a two-dimensional image $\bfX_i$ on the $m\times m$ grid $\{1,2,\ldots,m\}^2$ with $m=30$.  The covariates are generated following two covariance structures: exponential (``Exp'') and with shared structure (``SS'')  with the signal, $\bfbeta$.  The exponential covariates are Gaussian with mean $E(X_{ij})=0$ and cov$(X_{i,j},X_{i,l}) = \exp(-d_{j,l}/\vartheta_X)$, where $d_{j,l}$ is the distance between locations $j$ and $l$ and $\vartheta_X$ controls the range of spatial dependence.  The covariates generated with shared structure with $\bfbeta$ are $\bfX_i = {\tilde \bfX}_i/2 + e_i\bfbeta$, where ${\tilde \bfX}_i$ is Gaussian with exponential covariance with $\vartheta_X=3$ and $e_i \sim \mN(0,\upsilon^2)$; this is denoted as ``SS($\upsilon$)''. The response is then generated as $Y_i\sim$ N$(\bfX_i^{\rT}\bfbeta,\sigma^2)$.  Both $\bfX_i$ and $Y_i$ are independent for $i=1,\ldots,n$.  We consider two true $\bfbeta$ images (``Five peaks'' and ``Triangle'', plotted in Figure \ref{f:beta_sim}), sample sizes $n\in\{100,250\}$, spatial correlation $\vartheta_X\in\{3,6\}$, and error standard deviation $\sigma\in\{2,5\}$.  For all combinations of these parameters considered we generate $S=100$ datasets. 

\begin{figure}
	\caption{True $\bfbeta$ images used in the simulation study.}\label{f:beta_sim}
	\begin{center}
	\includegraphics[width=6in,height=2.7in]{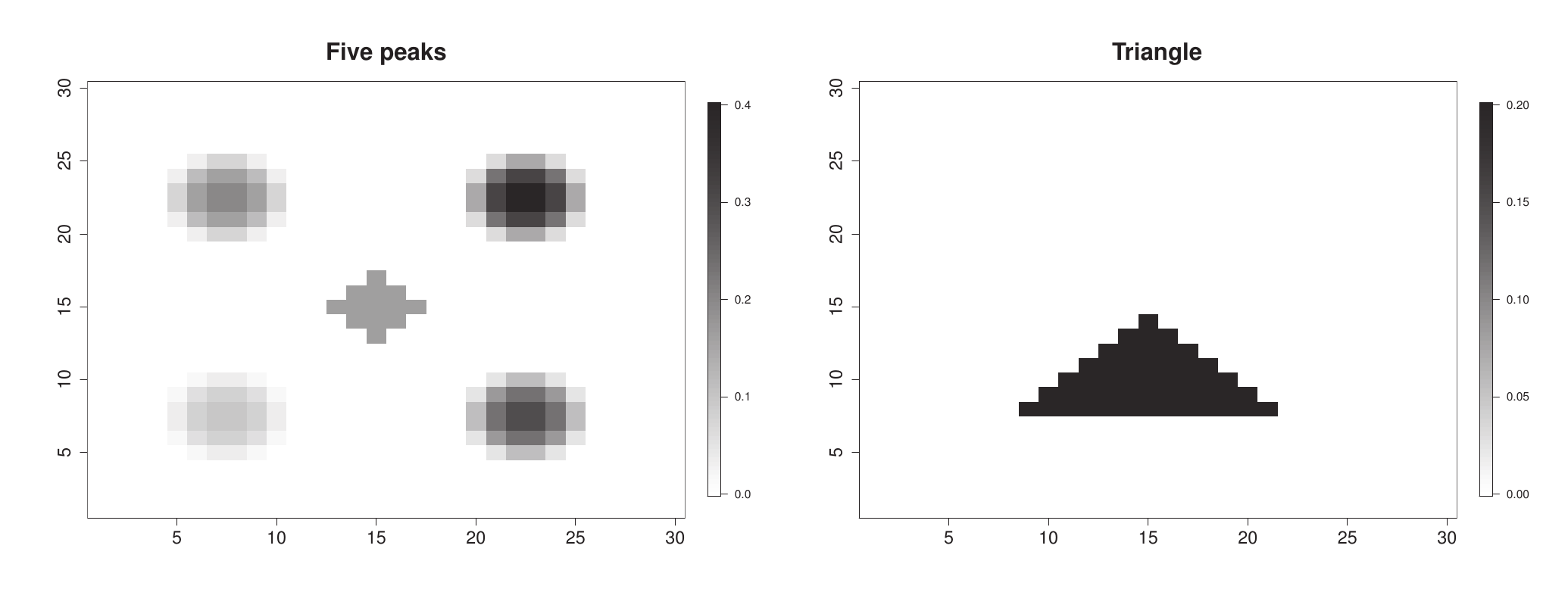}
	\end{center}
\end{figure}

\subsection{Methods}

We fit our model with a $m/2 \times m/2$ equally-spaced grid of knots covering $[1,m]\times [1,m]$ with bandwidth $\sigma_h$ set to the minimum distance between knots. We fit the model both with $\lambda>0$ and thus sparsity (``STGP'') and with $\lambda=0$ and thus no sparsity (``GP'').  For both models, we run the proposed Markov chain Monte Carlo algorithm 5,000 iterations with 1,000 burn-in, and compute the posterior mean of $\bfbeta$. For the sparse model, we compute the posterior probability of a nonzero $\beta(\bfs)$. 

 \def\hatbfbetaL{\hat{\beta}_{\mathrm{L}}^{\mathrm{v}}}
  \def\hatbfbetaFL{\hat{\beta}_{\mathrm{FL}}^{\mathrm{v}}}

We compare our method with the lasso \citep{tibshirani1996regression} and fused lasso \citep{Tibs:2005,Tibs:2011} penalized regression estimates
\beqn\label{fused}
  \hatbfbetaL &=& 
   \underset{\bfbeta}{\operatorname{argmin}}\left\{ (\bfY - \bfX\bfbeta)^{\rT}(\bfY-\bfX\bfbeta) + \tilde\lambda\sum_j|\beta(\bfs_j)|\right\}, \\
  \hatbfbetaFL &=& 
  \underset{\bfbeta}{\operatorname{argmin}}\left\{ (\bfY - \bfX\bfbeta)^{\rT}(\bfY-\bfX\bfbeta) +\tilde\lambda\sum_{j\sim k}|\beta(\bfs_j)-\beta(\bfs_k)| + \tilde\gamma\tilde\lambda\sum_j|\beta(\bfs_j)|\right\}.\nonumber
\eeqn 
The lasso estimate $\hatbfbetaL$ is computed using the {\tt lars} package \citep{lars} in {\tt R} \citep{R} and the tuning parameter $\tilde\lambda$ is selected using the Bayesian information criteria.  The fused lasso estimate $\hatbfbetaFL$ is computed using the {\tt genlasso} package \citep{genlasso} in {\tt R} and the tuning parameters $\tilde\gamma$ and $\tilde\lambda$ are selected using the Bayesian information criteria.  Due to computational considerations, we search only over $\tilde\gamma\in\{1/5,1,5\}$.

We also compare with a functional principle components analysis approach (``FPCA'').  We smooth each image using the technique of \cite{Xiao:2013} implemented in the {\tt fbps} function in {\tt R}'s {\tt refund} package \citep{refund}, compute the eigen decomposition of the sample covariance of the smoothed images, and then perform principal components regression using the lasso penalty tuned via the Bayesian information criteria.  We use the leading eigenvectors that explain 90\% of the variation in the sample images.  

Finally, we compare with the Bayesian spatial model of \cite{Goldsmith:2014} (``Ising'').  \cite{Goldsmith:2014} use the model $\beta(\bfs_j) = \tilde\alpha_j\theta_j$, where $\tilde\alpha_j\in\{0,1\}$ is the binary indicator that location $j$ is included in the model, and $\theta_j\in \mbR$ is the regression coefficient given that the location is included.  Both the $\tilde\alpha_j$ and $\theta_j$ have spatial priors; the continuous components $\theta_j$ follow a conditional autoregressive prior, and the binary components $\alpha_j$ follow an Ising (autologistic) prior \citep{handbook:2010} with full conditional distributions
 \beq\label{autologistic}
  \mathrm{logit}\{\Pi(\tilde\alpha_j=1|\tilde\alpha_l, l\ne j)\}= a + b\sum_{l\sim j}\tilde\alpha_l.
 \eeq
Estimating $a$ and $b$ is challenging because of the complexity of the Ising model \citep{Moller:2006}, therefore \cite{Goldsmith:2014} recommend selecting $a$ and $b$ using cross validation over $a\in (-4,0)$ and $b\in(0,2)$.  Due to computational limitations we select values in the middle of these intervals and set $a=-2$ and $b=1$.  Similar to our approach, 5,000 Markov chain Monte Carlo samples are simulated for the Ising model, and the first 1,000 are discarded as burn-in, and the posterior mean of $\beta(\bfs)$ and the posterior probability of a nonzero $\beta(\bfs)$ are computed. 
   	
\subsection{Results}

Table \ref{t:sim} gives the mean squared error for $\bfbeta$ estimation (averaged over location), type I error and  power for detect non-zero signals along with computing time.  The soft-thresholded Gaussian process (STGP) model gives the smallest mean squared error when the covariate has exponential correlation.  Compared to the Gaussian process (GP) model, adding thresholding reduces mean squared error by roughly 50\% in many cases.  As expected the functional principal component analysis (FPCA) methods gives the smallest mean squared error in final two scenarios where the covariates are generated to have a similar spatial pattern as the true signal.  Even in this case, the proposed method outperforms the other methods that do not exploit this shared structure.  For variable selection results, we only compare the proposed method with Fused lasso and the Ising model for a fair comparison, because Lasso does not incorporate spatial locations and other methods do not perform variable selection directly. The results show that Fused lasso has much larger Type I errors in all cases and the Ising model has a very small power to detect the signal in each case. It is clear that the proposed method is much better than Fused lasso and the Ising model for variable selection accuracy.  For the computing time, the proposed method is comparable to Fused lasso and faster than the Ising model.

\begin{table}
\small
\caption{Simulation study results.  Methods are compared in terms of mean squared error for $\bfbeta$ (``MSE for $\bfbeta$''), Type I error (\%) and Power (\%) for feature detection along with CPU time (minutes).  Data are generated for two true $\bfbeta_0$ (Fig. \ref{f:beta_sim}), covariance of the covariate $\bfX_i$ (exponential, ``Exp($\vartheta_X$)'', and shared structure, ``SS($\upsilon$)''), error standard deviation ($\sigma)$, and sample size ($n$). Results are reported as the mean over the $S$ simulated datasets.}\label{t:sim}

(a) MSE (multiplied by 1000) for $\beta$
	\begin{center}\begin{tabular}{cccc|cccccc}
  True $\bfbeta$ & cov ($\vartheta_X$) & $\sigma$ & $n$ & 
			Lasso & Fused lasso & FPCA & Ising & GP & STGP \\\hline
	Five peaks & Exp(3) & 5 & 100 &
	 31$\cdot$90   &    18$\cdot$48 &    3$\cdot$67 &    4$\cdot$44 &    2$\cdot$63 &     1$\cdot$65  \\		
	 & Exp(6) & 5 & 100 &
	 54$\cdot$99  &    2$\cdot$66  &    3$\cdot$33 &    4$\cdot$14 &    2$\cdot$07  &     1$\cdot$93  \\		
	 & Exp(3) & 2 & 100 &
	10$\cdot$20   &    4$\cdot$42  &    2$\cdot$51 &    2$\cdot$71 &    1$\cdot$50  &     0$\cdot$70   \\ 		
	 & Exp(3) & 5 & 250 &
	66$\cdot$85  &    1$\cdot$54  &    3$\cdot$01 &    5$\cdot$09 &    1$\cdot$71 &     0$\cdot$91  \\  	
	Triangle & Exp(3) & 5 & 100 &
	28$\cdot$31  &    18$\cdot$08 &    1$\cdot$83 &    2$\cdot$75 &    1$\cdot$80  &     0$\cdot$82  \\ 	
	 & Exp(6) & 5 & 100 &
	51$\cdot$90   &    4$\cdot$32  &    1$\cdot$63 &    2$\cdot$64 &    1$\cdot$76 &     0$\cdot$88  \\ 
	 & Exp(3) & 2 & 100 &
	7$\cdot$10    &    3$\cdot$74  &    1$\cdot$26 &    1$\cdot$35 &    1$\cdot$01 &     0$\cdot$55  \\ 		
	 & Exp(3) & 5 & 250 &
	65$\cdot$12  &    0$\cdot$69  &    1$\cdot$50  &    3$\cdot$33 &    1$\cdot$19 &     0$\cdot$68  \\ 	
	Triangle			& SS(2) & 5 & 100 &
    105$\cdot$80  &    70$\cdot$65 &    0$\cdot$98 &    2$\cdot$77 &    3$\cdot$28 &     1$\cdot$40   \\ 	
				& SS(4) & 5 & 100 &
   106$\cdot$62 &    71$\cdot$23 &    0$\cdot$34 &    3$\cdot$18 &    3$\cdot$39 &     1$\cdot$81  \\  
    	\end{tabular}\end{center}

(b) Type I error (\%)
\begin{center}\begin{tabular}{lll|ccccc}
cov &  True $\bfbeta$  & $n$ & Fused lasso & Ising & STGP \\\hline
Exp  &  Five peaks  &  100  &   18$\cdot$73 &   0$\cdot$09 &      3$\cdot$61 & \\
    	&              &  250  &   25$\cdot$88 &   0$\cdot$17 &      5$\cdot$62 & \\
      &  Triangle    &  100  &   19$\cdot$63 &   0$\cdot$06 &     3$\cdot$09 & \\
      &              &  250  &   11$\cdot$88 &   0$\cdot$14 &      4$\cdot$45 & \\
 SS   &  Five peaks  &  100  &   19$\cdot$58 &   0$\cdot$00 &    0$\cdot$39 & \\
      &              &  250  &   15$\cdot$57 &   0$\cdot$03 &     0$\cdot$71 & \\
      &  Triangle    &  100  &   20$\cdot$18 &   0$\cdot$00 &     1$\cdot$36 & \\
      &              &  250  &    1$\cdot$38 &   0$\cdot$03 &     2$\cdot$14 & 		
\end{tabular}\end{center}
		
(c) Power (\%)
\begin{center}\begin{tabular}{lll|ccccc}
True $\bfbeta$ & cov  & $n$ & Fused lasso  & Ising &STGP \\\hline
Exp  &  Five peaks  &  100  &   35$\cdot$21 &   4$\cdot$41 &    44$\cdot$78 & \\
     &              &  250  &   76$\cdot$45 &   9$\cdot$76 &   71$\cdot$77 & \\
     &  Triangle    &  100  &   49$\cdot$84 &   7$\cdot$71 &   89$\cdot$22 & \\
     &              &  250  &   93$\cdot$90 &  15$\cdot$84 &    96$\cdot$63 & \\
SS   &  Five peaks  &  100  &   29$\cdot$23 &   5$\cdot$59 &    30$\cdot$76 & \\
     &              &  250  &   49$\cdot$01 &   7$\cdot$52 &    48$\cdot$74 & \\
     &  Triangle    &  100  &   37$\cdot$86 &   7$\cdot$02 &    75$\cdot$53 & \\
     &              &  250  &   84$\cdot$27 &  12$\cdot$57 &    87$\cdot$14 & 
\end{tabular}\end{center}
		
(d) Computing time (minutes)
\begin{center}\begin{tabular}{llll|cccccc}
		True $\bbeta$ & cov ($\vartheta_X$) & $\sigma$ & $n$ & 
		Lasso & Fused lasso & FPCA & Ising & GP & STGP \\\hline
		Five peaks & 3 & 5 & 100 &
		0$\cdot$02 &    16$\cdot$77 &    5$\cdot$40  &    27$\cdot$61 &    4$\cdot$81 &     17$\cdot$69 \\ 
		\end{tabular}\end{center}

\end{table}

%\section{Analysis of EEG data}\label{s:eeg}
%
%To illustrate the proposed methods, we analysis the EEG data originally analyzed by XXX and subsequently by XXX.  DETAILS.  The point-wise means and standard deviations of the images are plotted in Figure \ref{f:meansd}.
%
%\begin{figure}
%	\caption{Point-wise mean and standard of the image covariates for the EEG data.}\label{f:meansd}
%%	\begin{center}\begin{picture}(400,400)
%%		\includegraphics[height=5.5in,width=5.5in]{MeanSD}
%%		\end{picture}\end{center}
%\end{figure}
%
%
%
%\subsection{Model comparisons}\label{s:comparisons}
%
%We compare the proposed method with the methods used in the simulation study with the exception of the fused lasso because code is not available for binary data.  For the LASSO and fPCA methods we used the {\tt R} package {\tt glmnet} \citep{glmnet}.   The Ising model of \cite{Goldsmith:2014} was modified to relate the image predictor to the binary response using a probit link as described in Section \ref{s:MCMC}.   Because the eigen decomposition required by the PCA methods was cumbersome for these large images, for model comparison we used every other column to give a XXX $\times$ XXX image for each subject; the full images where used in the final model fit in Section \ref{s:results}. Methods are compared using five-fold cross-validation; the area under the ROC curves are plotted in Figure \ref{f:auc}. 

\section{Analysis of EEG data}\label{s:eeg}

Our motivating application is the study of the relationship between the electrical brain activity as measured through multi channel electroencephalographic (EEG) signals and genetic predisposition to alcoholism. EEG is a medical imaging technique that records the electrical activity in the brain by measuring the current flows produced when the neurons are activated. The study comprises a total of 122 subjects - 77 alcoholic subjects and 45 non-alcoholic controls. For each subject 64 electrodes were placed on their scalp and EEG was recorded from each electrode at a frequency of 256Hz. The electrode positions were located at standard sites (standard electrode position nomenclature; American Electroencephalographic Association (1991)). %; for convenience, the map is included in the Supplementary material. 
The subjects were presented with 120 trials under several settings involving one stimulus or two stimuli. We consider the multichannel average EEG across the 120 trials corresponding to a single stimulus. The dataset is publicly available at the University of California at Irvine Knowledge Discovery of Datasets \underline{ https://kdd.ics.uci.edu/databases/eeg/eeg.data.html}. 

These data have been previously analyzed by \cite{li2010dimension,hung2013matrix} and \cite{zhou2014regularized}; however all the existing literature ignored the spatial location of the electrodes on the scalp and used instead their ID number, ranging from $1$ to $64$ which is assigned arbitrarily relative to the electrodes' position on the scalp. Our goal of the analysis is to detect the regions of brain which are most predictive of the alcoholism status; thus accounting for the actual position of the electrodes is a key component in our approach. In the absence of more sophisticated means to determine the electrodes' position on the scalp, we consider a lattice design and assign a two-dimensional location to each electrode that matches closely the electrode's standard position. Using the labels of the electrodes, we were able to identify only $60$ of them. As a result our analysis will be based on the multichannel EEG from these $60$ electrodes.

In accordance with the notation employed earlier, let $Y_i$ be the alcoholism status indicator with $Y_i=1$ if the $i$th subject is alcoholic and $0$ otherwise. Furthermore, let $X_i =\{X_i(\bfs_j; t): \bs_j \in \mbR^2, j=1, \ldots, 60, t=1, \ldots 256\}$ be the EEG image data for the $i$th subject which is indexed by a two-dimensional index accounting for the spatial location (on the matching lattice design), $\bs_j$, and one-dimensional index for time, $t$. 
%Figure \ref{f:meansd} illustrates the sample means and point-wise standard deviation of the EEG images by ignoring for simplicity the spatial index and using instead the electrodes' ID ($y$-axis) and time ($x$-axis).  

We use a probit model to relate the alcoholism status and the multichannel EEG image: $Y_i \mid X_i, \beta \sim  \mathrm{Bernoulli}(p_i)$ and $ \Phi^{-1}(p_i) = \sum_{j=1}^{60}\sum_{ k=1}^{256} X_i(\bfs_j, t_k) \beta(\bfs_j, t_k)$. The spatially-temporally varying coefficient function $\beta$ quantifies the effect of the image on the response over time and is modeled using the soft-thresholded Gaussian process on spatial and temporal domain.   We select a $5\times 5$ square grid of spatial knots and 64 temporal knots, for a total of 1,600 three-dimensional knots.  We initially fit a conditional autoregressive model with a different dependence parameter $\vartheta$ for spatial and temporal neighbors \citep{reich2007}, but found that the convergence was slow and that the estimates of both the spatial and temporal dependence were similar. Thus, we elected to use the same dependence parameter for all neighbors.  Also, we consider an informative prior for the threshold $\lambda \sim \mathrm{Uniform}(1.43, 1.96)$; intuitively this choice corresponds to an {\em a priori} inclusion probability between 5\%-15\%.   

We evaluate the prediction performance of the proposed model. Figure \ref{f:auc} shows the receiving operating characteristic curve (ROC) using five-fold cross validation. The results are compared with the ones corresponding to the lasso, the functional principal component analysis and the Gaussian process approach (the soft-thresholding Gaussian process approach with thresholding parameter $\lambda=0$). To facilitate computation for these methods, we thin the time points by two, leaving 128 time points.  While no model is uniformly superior, the area under the curve (AUC) corresponding to our approach is optimal among the alternatives we considered.  

\begin{figure}
	\caption{ROC curves for the five-fold cross validation of the EEG data. AUC refers to area under the curve}\label{f:auc}
	\begin{center}
		\includegraphics[height=4in,width=4in]{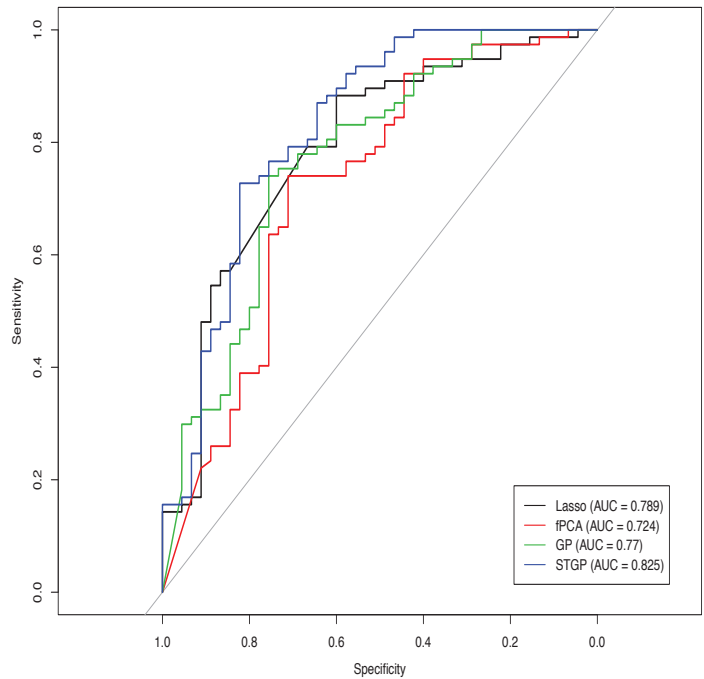}
	\end{center}
\end{figure}

The differences between the models are further examined in the estimated $\beta$ functions plotted in Fig. \ref{f:4fits} (for now we ignore the spatial location of the electrodes and plot them using their ID number). The lasso solution is non-zero for a single spatiotemporal location, while the functional principal component analysis and Gaussian process methods lead to non-sparse and thus uninterpretable $\beta$ estimates.   In contrast, the soft-thresholded Gaussian process based estimate is near zero for the  vast majority of locations, and isolates a subset of the electrodes near time point 86 as the most powerful predictors of alcoholism.

\begin{figure}
	\caption{Estimated $\bbeta$ for the EEG data.  The GP and STGP estimates are posterior means.}\label{f:4fits}
	\begin{center}
	
	\includegraphics[height=6in,width=6in]{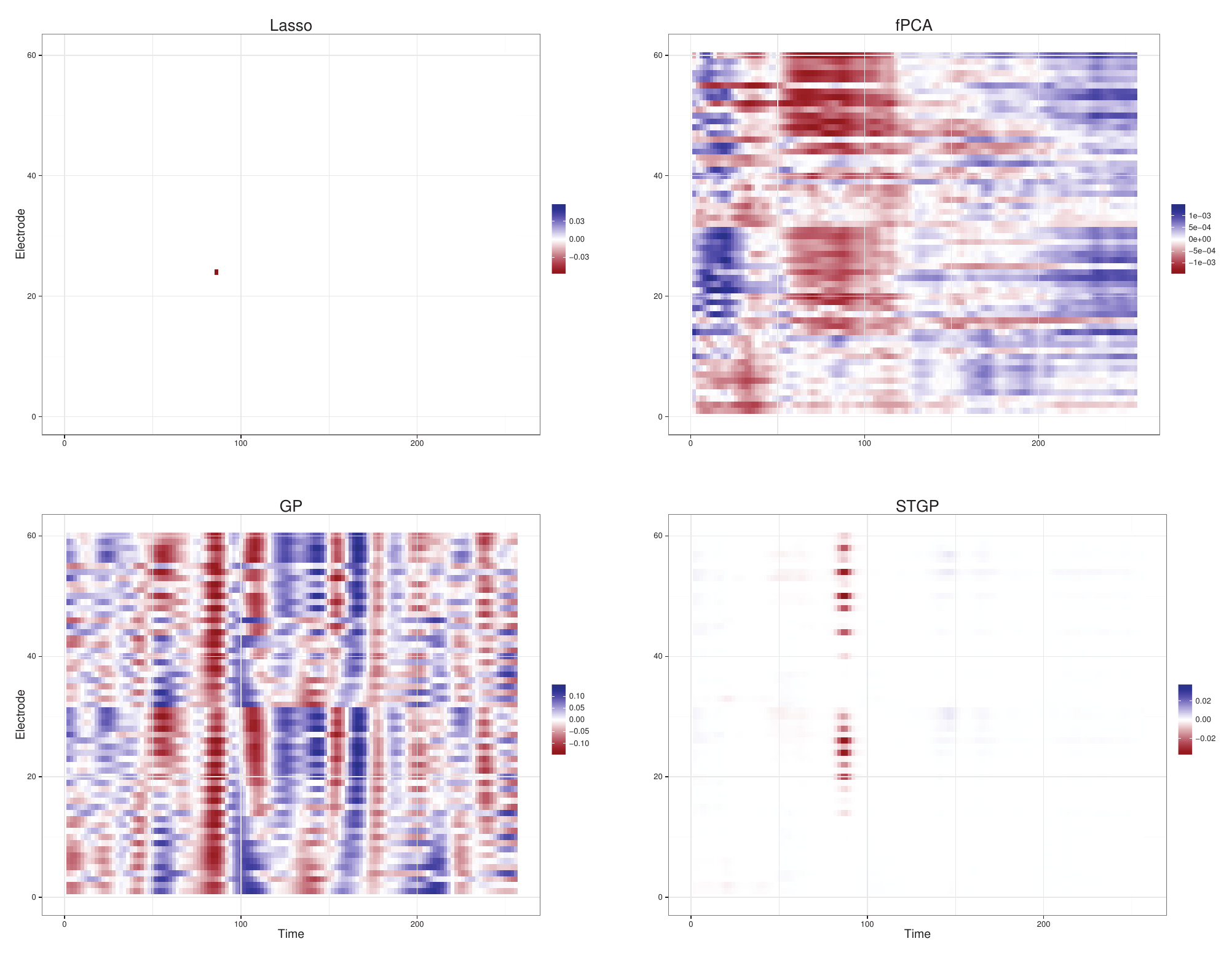}

	\end{center}
\end{figure}

Our analysis indicates that EEG measurements at time $t=86$, which roughly corresponds to $0.336$ fraction of second, are predictive of the alcoholism status. This observation is further confirmed by the plot of the posterior probability of non-zero $\beta(\bfs_j, t)$'s in Fig. \ref{f:finalfit}a. This implies a delayed reaction to the stimulus; though such finding has to be confirmed with the investigators. To gain more insight into these findings, Fig. \ref{f:finalfit}b-\ref{f:finalfit}d focus on a particular time and display the posterior mean and posterior probability of nonzero across the electrodes locations. They indicate that the right occipital/lateral part is the most predictive of the alcoholism status.

\begin{figure}
	\caption{Summary of STGP analysis of the EEG data.  Panel (a) plots the posterior probability of a nonzero $\beta(\bfs,t)$; each electrode is a line plotted over time $t$.  The remaining panels map either the posterior probability of a nonzero $\beta(\bfs,t)$ or the posterior mean of $\beta(\bfs,t)$ at individual time points.}\label{f:finalfit}
	\begin{center}

\includegraphics[height=6in,width=6in]{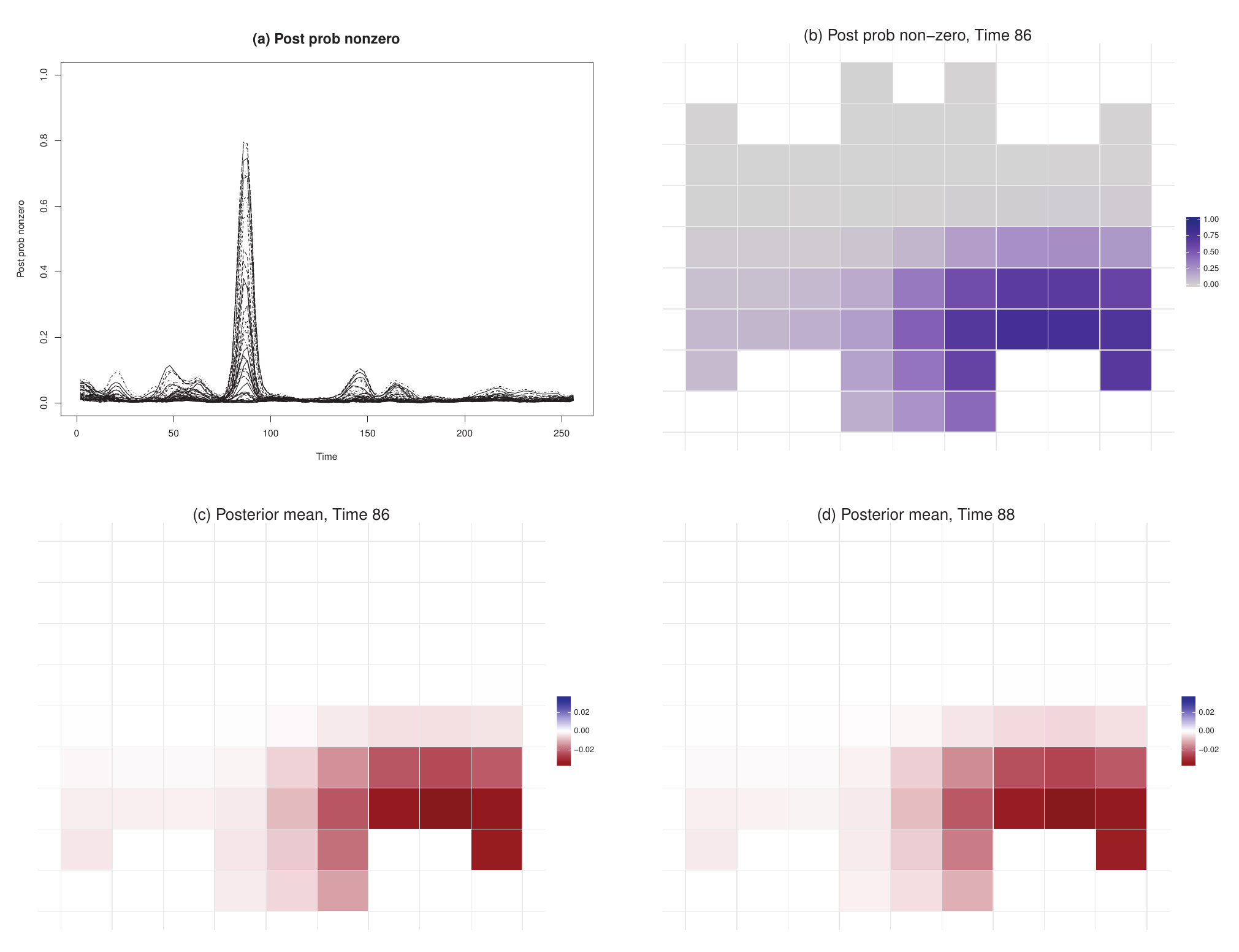}

     \end{center}
\end{figure}

\section{Discussions}\label{s:dis}
In this work, we proposed a new class of Bayesian nonparametric prior model: the soft-thresholded Gaussian process, for variable selection in the scalar-on-image regression. It is completely different from the hard thresholded Gaussian process developed by~\cite{shi2015thresholded} for the image-on-scalar regression. The soft-thresholded Gaussian process has two desired properties of Bayesian nonparametric models: 1) the prior support is large for a collection of piecewise smooth, sparse and continuous functions and 2) the posterior computation is feasible. We establish the posterior consistency for parameter estimation and variable selection for the normal response model. To the best of our knowledge, we are the first to obtain posterior consistency and the spatial variable selection consistency for scalar-on-image regression model.  The regularity conditions for the theoretical development are not strong, the results hold even when the number of true predictors is greater than sample size. Also, we develop efficient posterior computation algorithms using the low rank approximation through the conditional autoregressive model. The posterior approximation computation algorithm is general and can be used for other models besides the normal response model. Our simulation studies and the analysis of EEG data both indicate that the proposed approach performs better than all existing methods in terms of model parameter estimations, predictions and scientific findings. 

The proposed method generates a few feature directions that we consider to pursue. First, we plan to develop a more efficient posterior computation algorithm for analysis of voxel-level functional magnetic resonance imaging (fMRI) data, which typically contains 180,000 voxels for each subject. Any fast and scalable Gaussian processes approximation approach can be potentially applied to the soft-thresholded Gaussian process. For example, the recent ideas of nearest-neighbor Gaussian process approach by~\cite{datta2016hierarchical} can be applied to our model. Also, it is of great interest to perform joint analysis of dataset involving  multiple imaging modalities, such as fMRI, diffusion tensor imaging (DTI) and structural MRI. It is very changeling to model the dependence between the multiple imaging modality over space and to select the interactions between multiple modality imaging predictors in the scalar-on-image regression. The extension of the soft-thresholded Gaussian process can provide a potential solution to this problem. The basic idea is to introduce hierarchical latent Gaussian processes and different types of thresholding parameters for different modalities, leading to an hierarchical soft-thresholded Gaussian process as the prior model for the effects of interactions. 

\section*{Acknowledgement}
This research was supported partially by the National Institutes of Health grant R01 MH105561 (Kang), R01 NS085211 (Staicu) and R01 DE024984 (Reich), and the National Science Foundation grant DMS 0454942 (Staicu) and DMS 1513579 (Reich).

\section*{Supplementary material}
\label{SM}

Supplementary material contains the proofs of Lemmas 1, 3--5 and A1 -- A6. 
%Further material such as technical details, extended proofs, code, or additional  simulations, figures and examples may appear online, and should be briefly mentioned as Supplementary Material where appropriate.  Please submit any such content as a PDF file along with your paper, entitled `Supplementary material for Title-of-paper'.  After the acknowledgements, include a section `Supplementary material' in your paper, with the sentence `Supplementary material available at \Bka\ online includes $\ldots$', giving a brief indication of what is available.  Further instructions will be given when a paper is accepted.

\appendix

%\appendixone
\section*{Appendix 1}
\subsection*{Covering Number for Sieves}
\begin{lemma}\label{lem:covering_number}
 The $\epsilon$-covering number $N(\epsilon, \Theta_{n}, \|\cdot\|_{\infty})$ of $\Theta_{n}$ in the supremum norm satisfies
$$\log N(\epsilon, \Theta_{n}, \|\cdot\|_{\infty})\leq Cp_n^{1/(2\rho)} \epsilon^{-d/\rho}.$$
 \end{lemma}

\subsection*{Test Constructions}
\def\cV{\mathcal{V}}
\def\cU{\mathcal{U}}
\def\cW{\mathcal{W}}
\begin{lemma}\label{lem:gamma_low_bound}
 Suppose Condition \ref{cond:loc} holds for all $\bfs_j$ for $j = 1,\ldots, p_n$ and $K$ be the constant in Condition \ref{cond:loc}. Let $\upsilon>0$ be a constant. For each integer $n$, let $\Lambda_n$ be a collection of continuous functions, where each function $\gamma(\bfs)$ is differentiable on a set $\cD$ that is dense in $\cB$ and $ \sup_{\bfs \in \cD} |D^{\bftau}\gamma| \leq p_n^{\|\bftau\|_1/2d}+\upsilon$, for $\|\bftau\|_1 \geq 0$. For each function $\gamma\in \Lambda_n$ and $\epsilon > 0$, define $\cV_{\epsilon, \gamma} = \{\bfs: |\gamma(\bfs)| > \epsilon \}$. For all $n > N$ and $\gamma \in \Lambda_n$, 
$$\sum_{j=1}^{p_n} |\gamma(\bfs_j)| \geq \frac{\lambda(\cV_{\epsilon,\gamma})K \epsilon p_n}{2}. $$
\end{lemma}

\begin{lemma}\label{lem:beta_low_bound}
Suppose Conditions \ref{cond:num_cov_order} and \ref{cond:beta} hold. For each $\epsilon>0$, there exists an integer $N$ and $r>0$ such that, for all $n>N$ and for all $\beta \in \Theta_n$ such that $\|\beta -\beta_0\|_1 > \epsilon$, then 
$$\sum_{j=1}^{p_n} |\beta(\bfs_j) - \beta_0(\bfs_j)| > r p_n.$$
\end{lemma}

\begin{lemma}\label{lem:eta_low_bound}

For any $0<\epsilon<1$ and $0<r<\epsilon^2$,  let 
$$A_n = \left\{\sum_{i=1}^n p_n^{-1/2}\left|\sum_{j=1}^{p_n} X_{i,j}[\beta(\bfs_j)-\beta_0(\bfs_j)]\right| \geq n r\right\}.$$ There exists an integer $N$ and constant $D>0$ such that if for all $n>N$ and for all $\beta \in \Theta_n$,
\begin{eqnarray*}
\Pi\left[p_n^{-1/2}\left|\sum_{j=1}^{p_n} X(\bfs_j)[\beta(\bfs_j)-\beta_0(\bfs_j)]\right|>\epsilon \right] > \epsilon,
\end{eqnarray*}
then 
\begin{eqnarray*}
\Pi\left[A_n^C \right] \leq \exp(-D n) \quad \mbox{ and } \quad \Pi\left[\bigcup_{m=1}^{\infty}\bigcap_{n=m}^\infty A_m\right] = 1. 
\end{eqnarray*}
\end{lemma}

\begin{lemma}\label{lem:test}
Suppose $\bfalpha_0 = (\alpha_{0,1},\ldots, \alpha_{0,q})^{\rT}$ and $\sigma^2_0$ are known. The test statistic for the hypothesis testing problem
\begin{eqnarray*}
H_0: \beta = \beta_0 \in \Theta,\qquad \mbox{ and } \qquad H_1: \beta = \beta_1 \in \Theta.
\end{eqnarray*}
is give by
$$\Psi_n[\beta_0,\beta_1] = I\left[\sum_{i=1}^n \delta_i \left(\frac{Y_i - \eta_{i,0}}{\sigma_0}\right) > 2 n^{\upsilon+1/2}\right],$$
where 
\begin{eqnarray*}
\eta_{i,m} = \sum_{k=1}^q \alpha_{0,k} W_{i,k} + p_n^{-1/2}\sum_{j=1}^{p_n} \beta_m(\bfs_j) X_{i,j},
\end{eqnarray*}
for $m = 0,1$, $\delta_i = 2I[ \eta_{i,1} > \eta_{i,0}]-1$ and $\upsilon_0/2<\upsilon<1/2$. Then for any $r>0$, there exist constants $C_0$, $C_1$, $N$ and $r_0>0$ such that for any $\beta_0$ and $\beta_1$ satisfy
$\sum_{j=1}^{p_n} |\beta_{1}(\bfs_j) - \beta_0(\bfs_j)| > r p_n$, 
for any $n > N$, we have
$$\mE_{\beta_0}[\Psi_n(\beta_0,\beta_1)] \leq C_0\exp(-2 n^{2\upsilon}).$$
and for any $\beta$ with $\|\beta-\beta_1\|_{\infty} < r_0/\{4c_{\mmax}^{1/2}\}$, 
$$ \mE_{\beta}[1-\Psi_n(\beta_0,\beta_1)] \leq C_0\exp(- C_1 n ).$$
\end{lemma}

\section*{Appendix 2}
\subsection*{Proofs of Lemma \ref{lemma:latent_process} }
\begin{proof}
For any $\lambda_0 > 0$, set $\alpha(\bfs) = \beta_0(\bfs) + \lambda_0$ for $\bfs \in \overline\cR_1$ and $\alpha(\bfs) = \beta_0(\bfs) - \lambda_0$ for $\bfs \in \overline\cR_{-1}$. Then by condition (C1), $\alpha(\bfs)$ is smooth over $\overline \cR_1 \cup \overline \cR_{-1}$, i.e.
$$\alpha(\bfs)I[\bfs \in \overline \cR_1 \cup \overline \cR_{-1}] \in \mC^{\rho}(\overline \cR_1 \cup \overline \cR_{-1}).$$
Next, we define $\alpha(\bfs)$ on another closed subset of $\cB$.   Since $\cB$ is compact, $\partial\cR_k$ for $k = -1,1$ is also compact. For any $\epsilon > 0$ and each $\bft \in \cB$, define an open ball $B(\bft,r) = \{\bfs: \|\bft-\bfs\|_2<r\}$, where $\|\cdot\|_2$ is the Euclidean norm. For $k = -1,1$,   note that
$$ \partial \cR_k \subseteq \bigcup_{\bft\in \partial\cR_k} B(\bft,r),$$
Since $\partial\cR_1 \cup\partial\cR_{-1}$ is compact, there exist $\bft_l\in \partial \cR_1 \cup \partial \cR_{-1}$, for $1\leq l\leq L$, such that
$$  \partial \cR_{-1} \subseteq \bigcup_{l=1}^{L_0} B(\bft_l,r),\qquad   \partial \cR_1 \subseteq \bigcup_{l=L_0+1}^{L} B(\bft_l,r)$$

Let  $\cR^*_0(r) = \cR_0 -  \bigcup_{l=1}^{L} B(\bft_l,r)$, then $ \cR^*_0 \subseteq \cR_0 - \partial \cR_1 \cup \partial \cR_{-1}$. Note that $\cR_0 -\partial \cR_1 \cup \partial \cR_{-1}$ is a non-empty open set, $\cR^*_0(r)$ is its closed subset and $\cR^*_0(r)$ will increase as $r$ decreases. There exists an $r_0$,  $0<r_0<1$, such that $\cR^*_0(r_0) \neq \emptyset$ and  $\left(\bigcup_{l=1}^{L_0} B(\bft_l,r_0)\right)\cap \left(\bigcup_{l=L_0}^{L} B(\bft_l,r_0)\right) = \emptyset$. The latter fact is due to $\cR_1 \cap \cR_{-1} = \emptyset$.  Since $\cR_{1}\cup\cR_{-1}$ is bounded and $\alpha\in\mC^{\rho}(\cR_{1}\cup\cR_{-1})$, then $M = \max_{0<\|\bftau\|_1\leq \rho} \sup_{\bft\in \cR_1\cup \cR_{-1}}|D^{\bftau}\alpha(\bft)| < \infty$. Take $r = \min\left\{\lambda_0/\{2M(\rho+1)^d+1\},r_0\right\}$.  Define $\alpha(\bfs) = 0$ if $\bfs \in \cR^*_0(r)$. Then $\alpha(\bfs)$ is well defined on a closed set $\cR^* = \cR^*_0 \cup \overline \cR_1 \cup \overline \cR_{-1}$, where $\cR^*_0 = \cR^*_0(r)$. 

Define a function 
\begin{eqnarray}\label{eq:phi}
\phi(\bfs,\bft)  &=& \sum_{\|\bftau\|_1\leq \rho} \frac{D^{\bftau}\alpha (\bft)}{\bftau!}(\bfs - \bft)^{\bftau}\\
&=& \alpha(\bft) + \sum_{0<\|\bftau\|_1\leq \rho} \frac{D^{\bftau}\alpha (\bft)}{\bftau!}(\bfs - \bft)^{\bftau}
\end{eqnarray}
when $\bft \in \partial \cR_1 \cup \partial \cR_{-1}$ and $\bfs \in  B(\bft,r_0)$, 
$$|\phi(\bfs,\bft)| \leq |\lambda_0| + \bigg|\sum_{0<\|\bftau\|_1\leq \rho} \frac{D^{\bftau}\alpha (\bft)}{\bftau!}(\bfs - \bft)^{\bftau}\bigg | \leq \lambda_0 + 2 M (\rho+1)^d r_0 $$
when $\bft \in \partial \cR^*_0(r_0)$ and $\|\bfs-\bft\| < r_0$, 
$$|\phi(\bfs,\bft)|\leq   2M (\rho+1)^d r_0.$$

Define 
\begin{eqnarray*}
\psi(\bft) = \sum_{l=1}^L \psi(\bft,\bft_l). 
\end{eqnarray*}
where 
$$
\psi(\bft,\bft_l) = C_l\exp\left\{-\frac{1}{1-\|\bft-\bft_l\|_2/r}\right\}I[\|\bft-\bft_l\|_2<r].
$$
We choose $C_l$ for $l = 1,\ldots, L$ such that
\begin{eqnarray*}
\sum_{l=1}^{L}\int_{B(\bft_l,r)}\psi(\bft,\bft_l) \md \bft = 1,\qquad \mbox{and} \sum_{l=L_0+1}^{L}\int_{B(\bft_l,r)}\psi(\bft,\bft_l)\md \bft < 1 - \frac{2M(\rho+1)^d r}{\lambda_0}. 
\end{eqnarray*}

We construct $\widetilde \beta_0(\bfs)$ by extending $\alpha(\bfs)$ from $\cR^*$ to the whole domain $\cB$. Let 
\begin{eqnarray}\label{eq:beta_0}
\widetilde\beta_0(\bfs) = 
\left\{\begin{array}{rl} \int_{\partial \cR^*} \phi(\bfs, \bft) \psi(\bft) \md \bft ,&\bfs \in \cB - \cR^*\\
\alpha(\bfs), & \bfs \in \cR^*
\end{array}\right.
\end{eqnarray}

Note that
\begin{eqnarray*}
&&\bigg|\int_{\partial \cR^*} \phi(\bfs, \bft) \psi(\bft) \md \bft\bigg|\\
&\leq&\sum_{l=1}^L \int_{\partial \cR^* \cap B(\bft_l,r)} \big| \phi(\bfs, \bft) \big| \psi(\bfs,\bft) \md \bft\\
&\leq&\sum_{l=1}^L \left(\int_{\partial \cR^*_0 \cap B(\bft_l,r)}\big| \phi(\bfs, \bft)\big| \psi(\bfs,\bft) \md \bft + \int_{\partial (\cR_1\cup \cR_{-1}) \cap B(\bft_l,r)} \big|\phi(\bfs, \bft) \big|\psi(\bfs,\bft) \md \bft\right)\\
&< & 2M(\rho+1)^d r (1-w_{1}) + (\lambda_0+2M(\rho+1)^d r) w_{1} < \lambda_0
\end{eqnarray*}
where
\begin{eqnarray*}
w_1 =  \sum_{l=L_0+1}^{L}\int_{B(\bft_l,r)}\psi(\bft,\bft_l)\md \bft < 1 -  \frac{2M(\rho+1)^d r}{\lambda_0}. 
\end{eqnarray*}

Next, we show that  
$$\lim_{\bfs\rightarrow \bfs_0} D^{\bftau} \widetilde \beta_0(\bfs) = D^{\bftau} \alpha(\bfs_0), \mbox{ for any } \bfs_0 \in \partial \cR^*\mbox{ and }\bftau \mbox{ with } \|\bftau\|_1\leq \rho. $$

For any $\epsilon$, $0<\epsilon<1$, since $D^{\bftau}\alpha$ is continuous over $\cR^*$, there exists some $\delta_1>0$, for all $\bft$ such that $\|\bft -\bfs_0\|<\delta_1$, we have $|D^{\bftau}\alpha(\bft) - D^{\bftau}\alpha(\bfs_0)|  < \epsilon/2$.  Take $\delta < \min\{\epsilon/\{2 (\rho+1)^d M\}, r, \delta_1\}$, as long as $\|\bfs_0 - \bfs\|< \delta$, we have
\begin{eqnarray*}
|D^{\bftau} \widetilde \beta_0(\bfs) - D^{\bftau} \widetilde \alpha(\bfs_0)| &\leq & \int_{\partial \cR^*}|D^{\bftau}\alpha(\bft)-D^{\bftau}\alpha(\bfs_0)|\psi(\bft)\md \bft \\
&&+ \sum_{\|\bftau'\|_1 \geq \|\bftau\|_1, \bftau'\neq \bftau} \int_{\partial \cR^*}\frac{|D^{\bftau'}\alpha(\bft)|}{\bftau'!}|(\bfs - \bft)^{\bftau'}| \psi(\bft) \md \bft \\
&<& \frac{\epsilon}{2} + (\rho+1)^d M \delta < \epsilon. 
\end{eqnarray*}

By condition 2 and the definition of $\alpha(\bfs)$, we have 
\begin{eqnarray}\label{eq:alpha}
\alpha(\bfs) = \lambda_0, \mbox{ for } \bfs \in \partial\cR_1, \qquad \mbox{and} \qquad \alpha(\bfs) = -\lambda_0,\mbox{ for } \bfs\in\partial\cR_{-1}
\end{eqnarray}
\end{proof}

\subsection*{Proof of Theorem \ref{thm:post_consistency}}
\begin{proof}
The proof can be done by verifying the conditions in Theorem A.1 of \cite{choudhuri2004bayesian}. Specifically, we have the condition on prior positivity of neighborhoods by Lemma \ref{lem:prior_pos}. By Lemma A\ref{lem:tail_prob}, Lemma \ref{lem:uniform_test} and Condition \ref{cond:num_cov_order}, as $n\to\infty$,
\begin{eqnarray*}
\mE_{\beta_0}\Psi_n &\to& 0,\\
\sup_{\beta \in \cU_{\epsilon}^{C}\cap \Theta_n}\mE_{\beta}[1-\Psi_n]&\leq& C_0 \exp(-C_1 n),\\
\Pi(\Theta_n^C) &\leq& K \exp(-b p_{n}^{1/d}) \leq K \exp(- C_3 n).
\end{eqnarray*}
This establishes the condition on the existence of tests. 
\end{proof}
\subsection*{Proof of Theorem 3}
\begin{proof}
Define $\cU_{\epsilon} = \{\beta\in\Theta: \|\beta-\beta_0\|_1<\epsilon\}$. 
Let $\cR_0 = \{\bfs:\beta_0(\bfs)=0\}$, $\cR_1 = \{\bfs:\beta_0(\bfs)>0\}$ and $\cR_{-1.} = \{\bfs:\beta_0(\bfs)<0\}$. 

For any $\cA \subseteq \cB$ and any integer $m\geq 1$, define 
$$\cF_m(\cA) = \left\{\beta\in\Theta: \int_{\cA} | \beta(\bfs) - \beta_0(\bfs) | \md \bfs<\frac{1}{m}\right\}.$$
Then $\cF_{m+1}(\cA)  \subseteq \cF_{m}(\cA)$ for all $m$ and $\cF_m(\cB) \subseteq \cF_{m}(\cA)$.

Consider
$$\cF_m(\cR_0)  = \left\{\beta\in\Theta: \int_{\cR_0} | \beta(\bfs) | \md \bfs<\frac{1}{m}\right\}.$$
By Theorem \ref{thm:post_consistency}, note that $\cU_{1/m} = \cF_m(\cB)$, thus,
$$\Pi(\cF_m(\cR_0) \mid \bfD_n) \geq \Pi(\cU_{1/m} \mid \bfD_n) \to 1,$$
 as $n\to \infty$ in $P_{\beta_0}^n$ probability.
 Also, 
 $$ \left\{\beta(\bfs) = 0,\ \mathrm{for\ all}\  \bfs \in \cR_0\right\} = \left\{\int_{\cR_0} | \beta(\bfs) | \md \bfs = 0\right\} = \bigcap_{m=1}^{\infty}\cF_m(\cR_0). $$
 By the monotone continuity of probability measure, we have
 \begin{eqnarray}\label{eq:zero_fp}
 \Pi\left\{\beta(\bfs) = 0 ,\ \mathrm{for\ all}\  \bfs \in \cR_0\mid \bfD_n\right\}  = \lim_{m\to\infty} \Pi(\cF_m(\cR_0) \mid \bfD_n) = 1.
 \end{eqnarray}
 as $n\to \infty$ in $P_{\beta_0}^n$ probability.
 
For any $\bfs_0 \in \cR_1$ and any integer $m\geq 1$, by Condition \ref{cond:beta}.3, there exists $\delta_0>0$, such that for any $\bfs_1\in B(\bfs_0,\delta_0) =\{\bfs: \|\bfs_1 - \bfs_0\|_1 <\delta_0\}$, such that
$$|\beta(\bfs_1) - \beta(\bfs_0)| < \frac{1}{2m}.$$
By Definition \ref{cond:beta}, $\cR_1$ is an open set, then there exists $\delta_1>0$, such that $B(\bfs_0,\delta_1) \subseteq \cR_1$. Taking $\delta = \min\{\delta_1,\delta_0\} > 0$, we have that
\begin{eqnarray*}
\lefteqn{\left\{\beta(\bfs_0) > -\frac{1}{m}, \ \mathrm{for\ all}\  \bfs_0\in\cR_1\right\}}\\
&\supseteq & \left\{\beta(\bfs_0) > \beta(\bfs_1) - \frac{1}{2m} \mbox{ and } \beta(\bfs_1)> - \frac{1}{2m}, \mbox{for some } \bfs_1 \in B(\bfs_0,\delta), \ \mathrm{for\ all}\  \bfs_0\in\cR_1\right\} \\
&\supseteq & \left\{ \int_{B(\bfs_0,\delta)}\beta(\bfs) d \bfs > -\frac{1}{2m},\ \mathrm{for\ all}\  \bfs_0\in\cR_1\right\}\\
&\supseteq & \left\{ \int_{B(\bfs_0,\delta)}\beta(\bfs) d \bfs > \int_{B(\bfs_0,\delta)}\beta_0(\bfs) d\bfs-\frac{1}{2m},\ \mathrm{for\ all}\  \bfs_0\in\cR_1\right\} \\
&\supseteq & \cF_{2m}[B(\bfs_0,\delta)] \supseteq \cU_{1/2m}
\end{eqnarray*}
 thus,
$$\Pi\left(\beta(\bfs_0) > -\frac{1}{m}, \ \mathrm{for\ all}\  \bfs_0\in\cR_1 \mid \bfD_n\right) \geq \Pi(\cU_{1/2m} \mid \bfD_n) \to 1,$$
 as $n\to \infty$ in $P_{\beta_0}^n$ probability.
 By the monotone continuity of probability measure, we have
 \begin{eqnarray}\label{eq:pos_fp}
 \Pi\left\{\beta(\bfs) > 0 ,\ \mathrm{for\ all}\  \bfs \in \cR_1\mid \bfD_n\right\}  = \lim_{m\to\infty} \Pi\left(\beta(\bfs_0) > -\frac{1}{m}, \ \mathrm{for\ all}\  \bfs_0\in\cR_1 \mid \bfD_n \right) \to 1.
 \end{eqnarray}
  as $n\to \infty$ in $P_{\beta_0}^n$ probability. Similar arguments can be made to show
 \begin{eqnarray}\label{eq:neg_fp}
 \Pi\left\{\beta(\bfs) < 0 ,\ \mathrm{for\ all}\  \bfs \in \cR_{-1}\mid \bfD_n\right\}  \to 1.
 \end{eqnarray}
   as $n\to \infty$ in $P_{\beta_0}^n$ probability. 
Combing \eqref{eq:zero_fp} -- \eqref{eq:neg_fp} completes the proof. 
\end{proof}

\bibliography{STGP}
\end{document}